\documentclass[11pt,reqno]{amsart} 

\usepackage{times}

\usepackage{amssymb}
\usepackage{amsxtra}
\usepackage{amsmath}

\usepackage{verbatim}
\usepackage{epsfig}
\usepackage{setspace}

\newtheorem{theorem}{Theorem}[section]
\newtheorem{definition}{Definition}[section]
\newtheorem{lemma}[theorem]{Lemma}
\newtheorem{proposition}[theorem]{Proposition}
\newtheorem{corollary}[theorem]{Corollary}
\newtheorem{example}{Example}[section]

\def\a{\alpha}

\def\gm{\gamma}
\def\d{\delta}

\def\k{\kappa}
\def\l{\lambda}

\def\G{\Gamma}

\def\F{{\mathbb F}}

\def\b{{\mathbf b}}
\def\c{{\mathbf c}}

\def\s{{\mathbf s}}

\def\x{{\mathbf x}}

\def\0{{\mathbf 0}}
\def\1{{\mathbf 1}}

\def\cC{{\mathcal C}}

\def\cG{{\mathcal G}}

\def\cM{{\mathcal M}}

\def\cS{{\mathcal S}}

\def\oJ{{\overline{J}}}
\def\oT{\overline{T}}

\def\mfB{{\mathfrak B}}
\def\mfC{{\mathfrak C}}
\def\mfG{{\mathfrak G}}
\def\mfP{{\mathfrak P}}
\def\mfR{{\mathfrak R}}
\def\mfT{{\mathfrak T}}

\def\TW{\mbox{\textsf{TW}}}

\def\tot{{\text{tot}}}
\def\tree{{\text{tree}}}
\def\path{{\text{path}}}
\def\vctree{{\text{vc-tree}}}
\def\vcpath{{\text{vc-path}}}

\def\bar{\overline}
\def\del{\setminus}
\def\to{\rightarrow}

\def\define{\stackrel{\mbox{\footnotesize def}}{=}}
\def\beq{\begin{equation}}
\def\eeq{\end{equation}}

\title{Constraint Complexity of Realizations of Linear Codes \\
on Arbitrary Graphs}
\thanks{This work was supported by a Discovery Grant from the 
Natural Sciences and Engineering Research Council (NSERC), Canada.}
\author{Navin Kashyap} 

\thanks{The author is with the Department of Mathematics and Statistics,
Queen's University, Kingston, ON K7L 3N6, Canada.
Email: \texttt{nkashyap@mast.queensu.ca}}

\renewcommand{\markboth}[2]
{\renewcommand{\leftmark}{\scshape{#1}}\renewcommand{\rightmark}{\scshape{#2}}}

\begin{document}

\renewcommand{\thefootnote}{\arabic{footnote}}
\setcounter{footnote}{0}
\begin{abstract}
A graphical realization of a linear code $\cC$ consists of an
assignment of the coordinates of $\cC$ to the vertices of a graph,
along with a specification of linear state spaces and linear
``local constraint'' codes to be associated with the edges and vertices,
respectively, of the graph. The $\k$-complexity of a graphical
realization is defined to be the largest dimension of any of its
local constraint codes. $\k$-complexity is a reasonable measure 
of the computational complexity of a sum-product decoding algorithm 
specified by a graphical realization. The main focus of this paper is on
the following problem: given a linear code $\cC$ and a graph $\cG$,
how small can the $\k$-complexity of a realization of $\cC$ on $\cG$ be?
As useful tools for attacking this problem, we introduce the Vertex-Cut
Bound, and the notion of ``vc-treewidth'' for a graph, which is 
closely related to the well-known graph-theoretic notion of treewidth.
Using these tools, we derive tight lower bounds 
on the $\k$-complexity of any realization of $\cC$ on $\cG$.
Our bounds enable us to conclude that good error-correcting codes can 
have low-complexity realizations only on graphs with large vc-treewidth.
Along the way, we also prove the interesting result that
the ratio of the $\k$-complexity of the best conventional trellis realization 
of a length-$n$ code $\cC$ to the $\k$-complexity of the 
best cycle-free realization of $\cC$ grows at most logarithmically
with codelength $n$. Such a logarithmic growth rate is, in fact, achievable.

\end{abstract}

\date{\today}
\maketitle



\section{Introduction\label{intro}}

The study of graphical models of codes and the associated
message-passing decoding algorithms is a major focus 
of current research in coding theory. This is attributable to 
the fact that coding schemes using graph-based iterative 
decoding strategies --- \emph{e.g.}, turbo codes and low-density 
parity check (LDPC) codes --- have low implementation complexity,
while their performance is close to the optimum predicted by theory.
A unified treatment of graphical models and the associated
decoding algorithms began with the work of Wiberg, Loeliger and Koetter 
\cite{wiberg},\cite{WLK95}, and has since been abstracted and refined 
under the framework of the generalized distributive law \cite{AM00},
factor graphs \cite{KFL01}, and normal realizations 
\cite{For01},\cite{For03}. In fact, the study of cycle-free 
graphical models of codes (\emph{i.e.}, models in which the 
underlying graphs are cycle-free) can be traced back to 
the introduction of the Viterbi decoding algorithm in the 1960's, 
which led to the study of trellis representations of codes. 
A comprehensive account of the history and development of 
trellis representations can be found in \cite{vardy}.

In this work, we will follow the approach of Forney 
\cite{For01},\cite{For03}, and Halford and Chugg \cite{halford} in 
studying the general ``extractive'' problem of constructing low-complexity 
graphical models for a given linear code. Roughly speaking, a low-complexity 
graphical model is one that implies a low-complexity decoding algorithm. 
In particular, in this paper, we investigate the question of 
how small the complexity of an arbitrary graphical model for 
a given code can be.

We briefly introduce graphical models here; a detailed
description can be found in Section~\ref{background_section}.
A graph decomposition of a code $\cC$ is a mapping of 
the set of coordinates of $\cC$ to the set of vertices of a graph.
A graph decomposition may be viewed as an assignment of symbol variables 
to the vertices of the graph. A graph decomposition can be extended 
to a graphical model which additionally assigns state variables 
to the edges of the graph, and specifies a local constraint code
at each vertex of the graph. The full behavior of the model is 
the set of all configurations of symbol and state variables 
that satisfy all the local constraints. Such a model is called 
a graphical realization of $\cC$ if the restriction of 
the full behavior to the set of symbol variables 
is precisely $\cC$. The realization is said to be cycle-free
if the underlying graph in the model has no cycles. A trellis 
representation of a code can be viewed as a cycle-free realization 
in which the underlying graph is a simple path. In contrast, 
a \emph{tailbiting} trellis representation \cite{KV02},\cite{KV03}
is a graphical realization in which the underlying graph consists
of a single cycle.

We will focus our attention on the case of realizations of linear codes 
on \emph{connected} graphs only. Indeed, there is no loss of 
generality in doing so, since a linear code $\cC$ has a 
realization on a graph $\cG$ that is not connected if and only 
if $\cC$ can be expressed as the direct sum of codes that may be
individually realized on the connected components of $\cG$ \cite{For01}. 
In this context, we will refer to cycle-free graphical realizations 
simply as tree realizations, as the underlying graph is a connected, 
cycle-free graph, \emph{i.e.}, a tree. 

It is by now well known that any graphical realization of a code 
specifies a canonical iterative message-passing decoding algorithm, 
namely, the sum-product algorithm, on the underlying graph 
\cite{AM00},\cite{For01},\cite{KFL01},\cite{wiberg}. When the underlying graph 
is a tree, the sum-product algorithm provides an exact implementation of 
maximum-likelihood (ML) decoding. Even when the underlying graph contains
cycles, empirical evidence suggests that, in many cases, the sum-product
algorithm continues to be a good approximation to ML decoding. 

The computational complexity of the sum-product algorithm associated
with a graphical realization of a code is largely determined by the sizes
of the local constraint codes in the realization. 
In Section~\ref{complexity_section} of this paper, we define 
various measures of ``constraint complexity'' of a graphical realization 
that can be used as estimates of the computational complexity of 
sum-product decoding. These complexity measures may be viewed as 
generalizations of previously proposed measures of trellis complexity 
\cite{vardy},\cite{KV03}, and tree complexity \cite{For03},\cite{halford}. 
However, for the most part, we focus on the \emph{$\k$-complexity} of 
a graphical realization, which we define to be the maximum of 
the dimensions of the local constraint codes in the realization. 

In the restricted context of tree realizations, it has previously 
been established that certain ``minimal'' tree realizations can 
be canonically defined. Let the term \emph{tree decomposition}
denote a graph decomposition in which the graph is a tree. 
It is known that among all tree realizations of a code $\cC$ 
that extend a given tree decomposition, there is one 
that minimizes the dimension of the state space at each edge
of the underlying tree, and this minimal tree realization
is unique \cite{For01}. It has further been shown \cite{K2}
that this unique minimal tree realization also minimizes 
(among all tree realizations extending the given tree decompositions)
the dimension of the local constraint code at each vertex of the tree.
In particular, it has the least $\k$-complexity among all such 
tree realizations. 

In contrast, there is very little known about the general case 
of realizations of a code on an arbitrary (not necessarily cycle-free) 
graph. For instance, there appear to be no ``canonical'' minimal realizations 
that can be defined in this situation. The only systematic study
in this direction remains that of Koetter and Vardy \cite{KV02},\cite{KV03}, 
who studied minimal tailbiting trellis representations of codes,
which, as already mentioned, are graphical realizations in which
the underlying graph consists of exactly one cycle. 
Beyond this basic (though by no means easy) case, there is little 
of interest in the literature on the complexity of realizations
of codes on arbitrary graphs, the notable exception to this being the
work of Halford and Chugg \cite{halford}. 

In their work, Halford and Chugg lay the foundations for a 
systematic study of complexity of graphical realizations. 
Their main result is the ``Forest-Inducing Cut-Set Bound'',
which gives a lower bound on the constraint complexity of a 
graphical realization in terms of its minimal tree complexity.
However, this bound does not appear be user-friendly in practice.
The main limitation of their approach is that they rely on the 
Edge-Cut Bound of Wiberg \emph{et al.} \cite{wiberg},\cite{WLK95}
to derive their results. While the Edge-Cut Bound has been put to good use 
in the study of ``state complexity'' of graphical realizations 
\cite{CFV99},\cite{For01}, it is of limited value in the analysis
of constraint complexity.

The main aim of our paper is to present useful and tight 
lower bounds on the constraint complexity of a graphical realization.
Our bounds also provide considerable insight into the problem of finding 
low-complexity graphical realizations. The fundamental tool in our analysis
is the Vertex-Cut Bound, which we state and prove in 
Section~\ref{cutset_bnd_section}. The Vertex-Cut Bound is
a natural analogue of the Edge-Cut Bound, but as we shall see,
it is more suitable for use in the analysis of constraint complexity.

In Section~\ref{vctree_section}, we define a data structure called
\emph{vertex-cut tree} that stores the information necessary about
a graph to effectively apply the Vertex-Cut Bound. Vertex-cut trees
are similar in structure to the junction trees associated with 
belief propagation algorithms \cite{jensen},\cite{AM00}. 
The \emph{vc-width} of a vertex-cut tree is a measure of the size 
of the vertex-cut tree, and the \emph{vc-treewidth} of a graph 
is the least vc-width among all its vertex-cut trees. 
The vc-treewidth of a graph is very closely related to
the notion of treewidth of graphs much studied in graph theory 
\cite{RS-I},\cite{bod93}. 

Using the Vertex-Cut Bound and the notion of vertex-cut trees,
we derive, in Section~\ref{lobnds_section}, a suite of lower bounds
on the $\k$-complexity of graphical realizations of a linear code $\cC$.
We state one of these bounds here as an illustrative example.
Let $\k_\tree(\cC)$ be the least $\k$-complexity among all 
tree realizations of $\cC$. Consider an arbitrary graph $\cG$, 
and let $\k_\vctree(\cG)$ denote its vc-treewidth. Then, 
the $\k$-complexity of any realization of $\cC$ on $\cG$ is
bounded from below by the ratio $\k_\tree(\cC)/\k_\vctree(\cG)$.

We further apply our methods to answer certain questions raised 
in \cite{K2}. Borrowing terminology from \cite{vardy}, for a code $\cC$,
let $b(\cC)$ denote the least edge-complexity of any trellis 
representation of $\cC$ or any of its coordinate permutations. 
In the language of our paper, $b(\cC)$ is the least $\k$-complexity
of any conventional trellis realization of $\cC$.
We show that for any linear code of length $n$, 
we have $b(\cC)/\k_\tree(\cC) = O(\log_2 n)$, and that this is the 
best possible estimate of the ratio, up to the constant implicit 
in the $O$-notation. This is used to extend a known lower bound 
\cite{LV95} on $b(\cC)$ in terms of the length $n$, dimension $k$ and 
minimum distance $d$ of $\cC$, to a lower bound on $\k_\tree(\cC)$.

Our lower bound on $\k_\tree(\cC)$ has an important implication. 
It shows that if $\mfC$ is a code family with the property that,
for each $\cC \in \mfC$, $\k_\tree(\cC)$ is bounded from above 
by a fixed constant, then either the dimension or the minimum distance 
of the codes in $\mfC$ grows sub-linearly with codelength. 
Thus, such code families are not good from a coding-theoretic perspective. 
We also prove a slightly more general result, which can be 
roughly interpreted as saying that a good error-correcting code 
cannot have a low-complexity realization on a graph with small vc-treewidth. 
So, for good codes, if low-complexity graphical realizations exist, 
then they must necessarily exist on graphs with large vc-treewidth.

Some concluding remarks are made in Section~\ref{conclusion}, and 
an example in support of a statement in Section~\ref{complexity_section}
is given in an appendix.

\section{Background and Notation\label{background_section}}

In this section, we provide the necessary background, and define the 
notation we use in the paper. We take $\F$ to be an arbitrary finite field. 
Given a finite index set $I$, we have the vector space 
$\F^I = \{\x = (x_i \in \F,\ i \in I)\}$. For $\x \in \F^I$ 
and $J \subseteq I$, the notation ${\x|}_J$ will denote the 
\emph{projection} $(x_i,\ i \in J)$. Also, for $J \subseteq I$, 
we will find it convenient to reserve the use of $\bar{J}$ 
to denote the set $\{i \in I:\ i \notin J\}$.

\subsection{Codes\label{code_section}}
A \emph{linear code} over $\F$, defined on the index set $I$, is a 
subspace $\cC \subseteq \F^I$. In this paper, the terms 
``code'' and ``linear code'' will be used interchangeably to mean
a linear code over an arbitrary finite field $\F$, unless explicitly
specified otherwise. The dimension, over $\F$, of $\cC$ will be denoted 
by $\dim(\cC)$. An $[n,k]$ code is a code of length $n$ and dimension $k$. 
If, additionally, the code has minimum distance $d$,
then the code is an $[n,k,d]$ code. 

Let $J$ be a subset of the index set $I$. The \emph{projection} of $\cC$ 
onto $J$ is the code $\cC|_J = \{{\c|}_J:\ \c \in \cC\}$, which is a 
subspace of $\F^J$. We will use $\cC_J$ to denote the \emph{cross-section}
of $\cC$ consisting of all projections $\c|_J$ of codewords $\c \in \cC$ 
that satisfy ${\c|}_{\bar{J}} = \0$. To be precise, 
$\cC_J = \{\c|_J:\ \c \in \cC, {\c|}_{\bar{J}} = \0\}$. Note that
$\cC_J \subseteq \cC|_J$. Also, since $\cC_J$ is isomorphic to
the kernel of the projection map $\pi: \cC \rightarrow {\cC|}_\oJ$
defined by $\pi(\c) = {\c|}_\oJ$, we have that 
$\dim(\cC_J) = \dim(\cC) - \dim({\cC|}_\oJ)$. As a consequence,
we see that if $J \subseteq K \subseteq I$, then 
$\dim(\cC_J) \leq \dim(\cC_K)$.

If $\cC_1$ and $\cC_2$ are codes over $\F$ defined on mutually disjoint 
index sets $I_1$ and $I_2$, respectively, then their \emph{direct sum} is the
code $\cC = \cC_1 \oplus \cC_2$ defined on the index set $I_1 \cup I_2$,
such that $\cC_{I_1} = {\cC|}_{I_1} = \cC_1$ and 
$\cC_{I_2} = {\cC|}_{I_2} = \cC_2$. This definition naturally extends
to multiple codes (or subspaces) $\cC_\alpha$, where $\alpha$ is a
code identifier that takes values in some set $A$. Again, it must be assumed
that the codes $\cC_\alpha$ are defined on mutually disjoint index sets
$I_\alpha,\ \alpha \in A$. The direct sum in this situation
is denoted by $\bigoplus_{\alpha \in A} \cC_\alpha$. 

\subsection{Graphs\label{graph_section}}

In this paper, we are primarily interested in graphs that are connected,
so any unqualifed use of the term ``graph'' should be taken to 
mean ``connected graph''. Let $\cG = (V,E)$ be a graph, 
where $V$ and $E$ denote its vertex and edge sets, respectively. 
To resolve ambiguity, we will sometimes denote the vertex and 
edge sets of $\cG$ by $V(\cG)$ and $E(\cG)$, respectively. 
Given a $v \in V$, the set of edges incident with $v$ will 
be denoted by $E(v)$.

For $X \subseteq E$, we define $\cG \del X$ to be the subgraph 
of $\cG$ obtained by deleting all the edges in $X$. If $X$ consists 
of a single edge $e$, then we will write $\cG \del e$ instead of 
$G \del \{e\}$. If $\cG \del X$ is disconnected, then $X$ is called 
an \emph{edge cut} of $\cG$. Similarly, for $W \subseteq V$, 
we define $\cG - W$ to be the subgraph of $\cG$ obtained by 
deleting all the vertices in $W$ along with all incident edges.
If $W$ consists of a single vertex $v$, then we will write $\cG - v$ 
instead of $G - \{v\}$. If $\cG - W$ is disconnected, 
then $W$ is called a \emph{vertex cut} of $\cG$.

A tree is a connected graph without cycles. 
Vertices of degree one in a tree are called \emph{leaves}, 
and all other vertices are called \emph{internal nodes}.
Note that any subset of the edges of a tree constitutes an edge cut,
and any subset of internal nodes constitutes a vertex cut of the tree.
If $e$ is an edge in a tree $T$, then we will denote by $T^{(e)}$
and $\oT^{(e)}$ the two components of $T \del e$. If $v$ is a vertex
of degree $\d$ in a tree $T$, then we will use $T_1^{(v)}, 
T_2^{(v)}, \ldots, T_{\d}^{(v)}$ to denote the components of 
$T - v$.

A \emph{path} is a tree with exactly two leaves 
(the end-points of the path). All internal nodes in a 
path have degree two. A \emph{simple cycle} is a 
connected graph in which all vertices have degree two.
An \emph{$n$-cycle}, $n \geq 3$, is a simple cycle with $n$ vertices.

\subsection{Graphical Realizations of Codes\label{realize_section}}


\begin{figure}
\epsfig{file=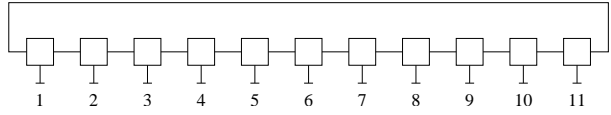, width=8cm}
\caption{A graph decomposition of $I = \{1,2,3,\ldots,11\}$.}
\label{11-cycle}
\end{figure}


The development in this section is based on the exposition of 
Forney \cite{For01},\cite{For03}; see also \cite{halford},\cite{K2}.
Let $I$ be a finite index set. A \emph{graph decomposition} of $I$ 
is a pair $(\cG,\omega)$, where $\cG = (V,E)$ is a graph, and 
$\omega: I \rightarrow V$ is an \emph{index mapping}. 
For a code $\cC$, we will usually write ``graph decomposition of $\cC$'' 
as shorthand for ``graph decomposition of the index set of $\cC$''. 
We again wish to emphasize that, unless explicitly stated otherwise, 
we will take $\cG$ to be a connected graph.
When $\cG$ is a tree, $(\cG, \omega)$ will be called a 
\emph{tree decomposition}. Pictorially, a graph decomposition $(\cG,\omega)$ 
is depicted as a graph with an additional feature: at each vertex $v$ 
such that  $\omega^{-1}(v)$ is non-empty, we attach special ``half-edges'', 
one for each index in $\omega^{-1}(v)$. Figure~\ref{11-cycle} 
depicts a graph decomposition of $I = \{1,2,3,\ldots,11\}$.


For a graph $\cG = (V,E)$, recall that $E(v)$, $v \in V$, 
denotes the set of edges incident with $v$ in $\cG$. 
Consider a tuple of the form\footnote{In referring to a tuple of 
the form $(\cG,\omega,(\cS_e,\ e \in E), (C_v,\ v \in V))$, we will
implicitly assume that $V$ and $E$ denote the vertex and edge sets,
respectively, of the graph $\cG$.}
$(\cG,\omega,(\cS_e,\ e \in E), (C_v,\ v \in V))$, where
\begin{itemize}
\item $(\cG,\omega)$ is a graph decomposition of $I$;
\item for each $e \in E$, $\cS_e$ is a vector space over $\F$ called
a \emph{state space};
\item for each $v \in V$, $C_v$ is a subspace of 
$\F^{\omega^{-1}(v)} \oplus\, \left(\bigoplus_{e \in E(v)} \cS_e\right)$,
called a \emph{local constraint code}, or simply, a \emph{local constraint}.
\end{itemize}
Such a tuple will be called a \emph{graphical model}. A graphical model
in which the underlying graph $\cG$ is a tree will be called a 
\emph{tree model}. The elements of any state space $\cS_e$ are called 
\emph{states}. The index sets of the state spaces $\cS_e$, $e \in E$, 
are taken to be mutually disjoint, and are also taken to be disjoint 
from the index set $I$ corresponding to the symbol variables. 

A \emph{global configuration} of a graphical model as above is
an assignment of values to each of the symbol and state variables.
In other words, it is a vector of the form 
$((x_i \in \F,\ i \in I), (\s_e \in \cS_e,\ e \in E))$.
A global configuration is said to be \emph{valid} if it satisfies all
the local constraints. Thus, 
$((x_i \in \F,\ i \in I), (\s_e \in \cS_e,\ e \in E))$ is a valid
global configuration if for each $v \in V$, $((x_i,\ i \in \omega^{-1}(v)),
(\s_e,\ e \in E(v))) \in C_v$. The set of all valid global configurations
of a graphical model is called the \emph{full behavior} of the model.

Note that the full behavior is a subspace 
$\mfB \subseteq \F^I \oplus\, \left(\bigoplus_{e \in E} \cS_e\right)$.
As usual, for $J \subseteq I$, ${\mfB|}_J$ denotes the projection of 
$\mfB$ onto the index set $J$. For future convenience, we also
define certain other projections of $\mfB$.
Let $\b = ((x_i,\, i \in I),\ (\s_e,\, e \in E))$
be a global configuration in $\mfB$. At any given $v \in V$,
the \emph{local configuration} of $\b$ at $v$ is defined as
$$
{\b|}_v = ((x_i,\, i \in \omega^{-1}(v)),\ (\s_e,\, e \in E(v))).
$$
The set of all local configurations of $\mfB$ at $v$ is then defined as
${\mfB|}_v = \{{\b|}_v:\, \b \in \mfB\}$. By definition, 
${\mfB|}_v \subseteq C_v$. Similarly, given any $e \in E$, 
if $\b$ is a global configuration as above, then we define ${\b |}_e =\s_e$;
we further define ${\mfB |}_e =\{{\b|}_e:\, \b \in \mfB\}$. 
Clearly, ${\mfB|}_e$ is a subspace of $\cS_e$. 

A graphical model $(\cG,\omega,(\cS_e,\ e \in E), (C_v,\ v \in V))$ is 
defined to be \emph{essential} if ${\mfB|}_v = C_v$ for all $v \in V$
and ${\mfB|}_e = \cS_e$ for all $e \in E$. When $\cG$ is a tree, there
is some redundancy in the above definition, as the
condition ${\mfB|}_e = \cS_e$ for all $e \in E$ actually implies
that ${\mfB|}_v = C_v$ for all $v \in V$ \cite[Lemma~2.2]{K2}.
It is worth noting that, in an essential graphical model,
at any edge $e = \{u,v\}$, the state space $\cS_e$, may be viewed 
as a projection of each of the local constraint codes $C_u$ and $C_v$.
Thus, for any $e \in E$, $\dim(\cS_e) \leq \dim(C_v)$, 
where $v$ is any vertex incident with $e$.
 
An arbitrary graphical model $\G$ with full behavior $\mfB$ can 
always be ``essentialized'' by simply replacing each 
local constraint $C_v$ in $\G$ with the projection ${\mfB|}_v$, 
and replacing each state space $\cS_e$ with the projection ${\mfB|}_e$. 
The resulting ``essentialization'' of $\G$ still has full behavior $\mfB$.

An essential graphical model $(\cG,\omega,(\cS_e,\ e \in E), (C_v,\ v \in V))$ 
is defined to be a \emph{graphical realization} of a code $\cC$, 
or simply a \emph{realization of $\cC$ on $\cG$}, if ${\mfB|}_I = \cC$.
A graphical realization of $\cC$ in which the underlying graph 
$\cG$ is a tree, is called a \emph{tree realization} of $\cC$. 
Our definition of a graphical (and tree) realization differs slightly
from the prior definitions in 
\cite{For01},\cite{For03},\cite{halford},\cite{K2}, 
in that we require the underlying graphical model to be essential. 
As explained above, any graph model can be essentialized, so there is no loss
of generality in this definition.

A graphical (resp.\ tree) realization 
$(\cG,\omega,(\cS_e,\ e \in E), (C_v,\ v \in V))$ of $\cC$
is said to \emph{extend}, or be an extension of, the graph 
(resp.\ tree) decomposition $(\cG,\omega)$ of $\cC$. We will
denote by $\mfR(\cC; \cG,\omega)$ the set of all graphical realizations
of $\cC$ that extend the graph decomposition $(\cG,\omega)$ of $\cC$.

Now, it is an easily verifiable fact that any tree decomposition
of a code can always be extended to a tree realization
of the code \cite{For03},\cite{K2}. Such an extension is not
unique in general, but we will describe a canonical ``minimal'' extension
a little later. More generally, any graph decomposition of a code $\cC$
can always be extended to a graphical realization of $\cC$, 
as we now explain.
Let $\cG = (V,E)$ be a connected graph, and suppose 
that $(\cG,\omega)$ is a graph decomposition of $\cC$. 
Take $T$ to be any spanning tree of $\cG$,
and let $E_T$ denote its edge set. Then, $(T,\omega)$ is a tree 
decomposition of $\cC$. As noted above, this tree decomposition can be 
extended to a tree realization  
$(T,\omega,(\cS_e,\ e \in E_T), (C_v,\ v \in V))$ of $\cC$. We further
extend this to a realization of $\cC$ on $\cG$ as follows. Define
the state spaces $\bar{\cS}_e$, $e \in E$, as
$$
\overline{\cS}_e = 
\begin{cases}
\cS_e & \text{if $e \in E_T$} \\
\{0\} & \text{if $e \in E \setminus E_T$},
\end{cases}
$$
and for each $v \in V$, define the local constraint 
$\bar{C}_v = C_v \oplus (\bigoplus_{e \in E(v) \setminus E_T} \{0\})$.
It should be clear that 
$(\cG,\omega,(\bar{\cS}_e,\, e \in E),\, (\bar{C}_v,\, v \in V))$
is a graphical realization of $\cC$.

Graphical realizations of codes in which the underlying graph is a 
path or a simple cycle have received considerable prior
attention in the literature. Such realizations were called 
``conventional state realizations'' 
(when the underlying graph is a path)
and ``tail-biting state realizations'' 
(when the underlying graph is a simple cycle) in \cite{For01}. 
We will call them ``trellis realizations''.  
Briefly, a \emph{trellis realization} of a code $\cC$ 
(defined on the index set $I$) is any extension of a 
graph decomposition of $\cC$ of the form $(\cG,\omega)$, 
where $\cG$ is either a path or a simple cycle, and $\omega$
is a surjective map $\omega:I \rightarrow V(\cG)$. 
A trellis realization in which the surjective map 
$\omega: I \rightarrow V(\cG)$ is not injective 
(so that $\omega$ is not a bijection), is usually
called a \emph{sectionalized} trellis realization. 
A \emph{conventional} trellis realization is one in which
the underlying graph $\cG$ is a path. When the 
underlying graph $\cG$ is a simple cycle, the trellis 
realization is said to be \emph{tailbiting}. The theory
of conventional trellis realizations is well established; 
see, for example, \cite{vardy}. On the other hand, 
tailbiting trellis realizations are less well understood;
the principal systematic study of these remains that of 
Koetter and Vardy \cite{KV02},\cite{KV03}. We remark that our 
requirement that graphical realizations have underlying
graph models that are essential corresponds to the requirement 
in \cite{KV02},\cite{KV03} that ``linear trellises'' be ``reduced''.

\section{Complexity Measures for Graphical 
Realizations\label{complexity_section}}

As observed in \cite{For01}, any graphical realization of a code 
specifies a class of associated graph-based decoding algorithms, 
namely, the sum-product algorithm and its variants.
Thus, ideally, any definition of a complexity measure for a 
graphical realization should try to capture the computational complexity 
of the associated decoding algorithms. The analysis 
in \cite[Section~V]{For01} shows that the computational complexity
of the sum-product algorithm specified by a given graphical realization 
of a code is determined in large part by the cardinalities,
or equivalently dimensions, of the
local constraint codes in the realization. Thus, as a simple
measure of the complexity of a graphical realization, which 
roughly reflects the complexity of sum-product decoding, we will
consider the maximum of the dimensions 
of the local constraint codes in the realization.


Let $\G = (\cG,\omega,(C_v,v \in V),(\cS_e,e\in E))$ be a graphical
realization of a code $\cC$. The \emph{constraint max-complexity},
or simply \emph{$\k$-complexity}, of $\G$ is defined to be 
$\k(\G) = \max_{v \in V} \dim(C_v)$. 
Now, recall that if $(\cG,\omega)$ is a graph decomposition of $\cC$,
then $\mfR(\cC;\cG,\omega)$ denotes the set of all graphical
realizations of $\cC$ that extend $(\cG,\omega)$.
We further define 
\beq
\k(\cC;\cG,\omega) = \min_{\G \in \mfR(\cC;\cG,\omega)} \k(\G)
\label{kCG_def}
\eeq
We add another level of minimization by defining, for a given
graph $\cG$, the \emph{$\cG$-width} of a code $\cC$ to be 
\beq
\k(\cC;\cG) = \min_{\omega} \k(\cC;\cG,\omega),
\label{kG_def}
\eeq
where the minimum is taken over all possible index mappings 
$\omega: I \to V(\cG)$, where $I$ is the index set of $\cC$. Thus,
the $\cG$-width of $\cC$ is the least $\k$-complexity of any
realization of $\cC$ on $\cG$, and may be taken to be a measure
of the least computational complexity of any sum-product-type decoding
algorithm for $\cC$ implemented on the graph $\cG$.  

A broader optimization problem of considerable interest is the following:
given a code $\cC$ and a family of graphs $\mfG$, identify a
$\cG \in \mfG$ on which $\cC$ can be realized with 
the least possible $\k$-complexity. We thus define 
\beq
\k(\cC;\mfG) = \min_{\cG \in \mfG} \k(\cC;\cG).
\label{kmfG_def}
\eeq
Two special cases of this definition --- 
treewidth and pathwidth (or trellis-width) --- 
are particularly of interest. We define treewidth first, and 
pathwidth a little further below. 

If we let $\mfT$ denote the set of all trees, then $\k(\cC;\mfT)$ is 
called the \emph{treewidth} of the code $\cC$ \cite{K2}, which we will denote 
by $\k_\text{tree}(\cC)$. The notion of treewidth 
(\emph{i.e.}, minimal $\k$-complexity among tree realizations) 
of a code was first considered by Forney \cite{For03},
and an analogous notion has been defined for matroids in \cite{HW06}.
The arguments in \cite[Section~V]{For03} (and also in \cite{HW06})
show that $\k_\text{tree}(\cC)$ can always be obtained by minimizing 
$\k(\cC;T,\omega)$ over tree decompositions $(T,\omega)$ 
in which $T$ is a cubic tree (\emph{i.e.}, a tree in which 
all internal nodes have degree 3), and $\omega$ is a bijection 
between the index set of $\cC$ and the set of leaves of $T$.

A complexity measure related to treewidth, termed minimal tree complexity, 
was defined and studied by Halford and Chugg \cite{halford}. 
Treewidth, as we have defined above, is an upper bound on the 
minimal tree complexity of Halford and Chugg.

The \emph{pathwidth}, $\k_{\text{path}}(\cC)$, of a code $\cC$ 
is defined to be the quantity $\k(\cC;\mfP)$, where $\mfP$ denotes 
the sub-family of $\mfT$ consisting of all paths. We will find it convenient
to refer to tree decompositions $(P,\omega)$, with $P \in \mfP$,
as \emph{path decompositions}. Thus, $\k(\cC;\mfP)$ is the minimum
value of $\k(\cC;P,\omega)$ as $(P,\omega)$ ranges over all
path decompositions of $\cC$. In fact, by the argument of 
\cite[Section~V.B]{For03}, the minimizing path decomposition $(P,\omega)$
may be taken to be one in which the index mapping $\omega$ is surjective.
Thus, $\k_{\text{path}}(\cC)$ is the least $\k$-complexity
of any conventional trellis realization of $\cC$, and so
we may also call it the \emph{(conventional) trellis-width}\footnote{For
this reason, what we have called $\k_{\text{path}}(\cC)$ here was called 
$\k_{\text{trellis}}(\cC)$ in \cite{K2}.} of $\cC$.
It is also known that sectionalization cannot reduce 
the $\k$-complexity\footnote{Our notion of $\k$-complexity
corresponds to the notion of ``edge-complexity'' in \cite{vardy}.}
of a trellis realization \cite[Theorem~6.3]{vardy}, and
hence, $\k_{\text{path}}(\cC)$ is the minimum
value of $\k(\cC;P,\omega)$ over all path decompositions $(P,\omega)$
in which the index mapping $\omega$ is a bijection between the index
set of $\cC$ and the vertices of $P$.

Measures of constraint complexity other than $\k$-complexity have been
proposed in the previous literature, especially in the context
of trellis realizations \cite{vardy},\cite{KV03}. 
The \emph{$\k^+$-complexity} of a graphical realization 
$\G = (\cG,\omega,(C_v,v \in V),(\cS_e,e\in E))$ is defined to be 
$\k^+(\G) = \sum_{v \in V} \dim(C_v)$. Note that 
$\frac{1}{|V|} \, \k^+(\G)$ is the average local constraint code
dimension in $\G$. On the other hand, it has been suggested 
\cite{For03} that the sum of the constraint code cardinalities
``may be a better guide to decoding complexity'' than 
$\k$-complexity or $\k^+$-complexity. Thus, we define
$\k^\tot(\G) = \sum_{v \in V} |C_v| = \sum_{v \in V} |\F|^{\dim(C_v)}$.
Analogous to (\ref{kCG_def})--(\ref{kmfG_def}), we may define
$\k^+(\cC,\cG,\omega)$, $\k^\tot(\cC;\cG,\omega)$, etc., 
but we will only touch upon these briefly in this paper.

We remark that while we have used constraint code dimensions
to define our complexity measures for graphical realizations, 
one could also define measures of complexity based on state-space 
dimensions. For example, we could define the state max-complexity,
$\sigma(\G)$, of a graphical realization $\G$ to be the 
maximum of the dimensions of the state spaces in $\G$. 
Similarly, we may consider the complexity measures $\sigma^+(\G)$
and $\sigma^\tot(\G)$ analogous to $\k^+(\G)$ and
$\k^\tot(\G)$. These measures are especially relevant and 
have been well studied in the context of trellis realizations; 
again, see \cite{vardy},\cite{KV03}. However, as noted by 
Forney \cite{For03}, measures of state-space complexity become 
less appropriate in the context of realizations on arbitrary graphs 
or trees. For example, for any code $\cC \subseteq \F^n$, 
one can always find a tree on which $\cC$ can be realized in such 
a way that all state-spaces have dimension at most 1. This would be the 
``star-shaped'' tree $T$ consisting of $n$ leaves connected
to a single internal node $v$ of degree $n$. Take $\omega$ to be
any bijection between the index set of $\cC$ and the leaves of $T$;
set $\cS_e = \F$ at each edge $e$ of $T$; and finally, 
take the local constraint code $C_v$ at the internal node to be $\cC$ 
itself, and take the local constraint codes at the leaves 
to be $[2,1]$ repetition codes. Clearly, the resulting 
tree model (after essentialization) is a tree realization of $\cC$.
Thus, it makes little sense to define a state-space complexity measure 
analogous to treewidth, unless we restrict the kind of trees on which we 
are allowed to realize the given code\footnote{We do get a reasonable
state-space analogue to treewidth if we restrict the class
of trees over which we attempt to minimize state-space complexity
to the class of cubic trees only; see \cite{K2}.}.

The astute reader may point out that in the trivial tree realization
above, the sum of the state-space dimensions is non-trivial, 
and so a state-space analogue to treewidth could potentially be
defined in terms of $\sigma^+$ or $\sigma^\tot$. This may be
true, but we do not pursue this further, since, as already observed
previously, complexity measures based on constraint code dimensions
are a better guide to decoding complexity. But, while on this topic,
we mention in passing that for any \emph{tree} realization $\G$
of a code $\cC$, it turns out that 
\beq
\k^+(\G) = \dim(\cC) + \sigma^+(\G).
\label{dim_sum_eq}
\eeq
Thus, the problem of minimizing $\sigma^+(\G)$ among
tree realizations $\G$ of a given code $\cC$ is equivalent to the 
problem of minimizing $\k^+(\G)$. The identity
in (\ref{dim_sum_eq}), which may be viewed as a generalization 
of the statement of Theorem~4.6 in \cite{KV03} 
for conventional trellis realizations, will not be proved here as it
would be an unnecessary deviation from the main line of our development.
It suffices to say that (\ref{dim_sum_eq}) follows from 
Theorem~3.4 in \cite{K2} by first verifying that it indeed holds
for any minimal tree realization $\cM(\cC;T,\omega)$, 
and then observing that the difference between 
$\k^+(\G)$ and $\sigma^+(\G)$ is preserved by
the state-merging process mentioned in the statement of that theorem.

Finally, we remark that the state max-complexity of
a graphical realization cannot exceed the constraint max-complexity
of the realization. This is because, as observed in 
Section~\ref{realize_section}, in any essential 
graphical model $(\cG,\omega,(\cS_e, e \in E), (C_v, v \in V))$, 
for each edge $e \in E$, we have $\dim(\cS_e) \leq \dim(C_v)$, where
$v$ is any vertex incident with $e$.

\subsection{Minimal Realizations\label{minimal_section}}

Given a code $\cC$ and a tree decomposition $(T,\omega)$ of $\cC$,
there exists a tree realization, 
$(T,\omega,(\cS_e^*,\ e \in E), (C_v^*,\ v \in V))$, of $\cC$ 
with the following property \cite{For01},\cite{For03}: 
\begin{quote}
if $(T,\omega,(\cS_e,\ e \in E), (C_v,\ v \in V))$ is
a tree realization of $\cC$ that extends $(T,\omega)$, then 
for all $e \in E$, $\dim(\cS_e^*) \leq \dim(\cS_e)$.
\end{quote}
This \emph{minimal} tree realization, which we henceforth denote by 
$\cM(\cC; T,\omega)$, is unique up to isomorphism\footnote{Graphical
realizations $(\cG,\omega,(C_v,v\in V),(\cS_e,e\in E))$ and 
$(\cG,\omega,(C_v',v\in V),(\cS_e',e\in E))$ of a code $\cC$ 
are said to be \emph{isomorphic} if, for each $v \in V$, $C_v$ and $C_v'$ 
are isomorphic as vector spaces, and for each $e \in E$, $\cS_e$ and
$\cS_e'$ are isomorphic as vector spaces. We do not distinguish 
between isomorphic graphical realizations.}. Constructions of 
$\cM(\cC;T,\omega)$ can be found in \cite{For01},\cite{For03},\cite{K2}.

It has further been shown \cite{K2} that not only does $\cM(\cC;T,\omega)$ 
minimize (among realizations in $\mfR(\cC;T,\omega)$) 
the state space dimension at each edge of $T$, but it also 
minimizes the local constraint code dimension at each vertex of $T$. 
More precisely, $\cM(\cC;T,\omega)$ also has the following property:
\begin{quote}
if $(T,\omega,(\cS_e,\ e \in E), (C_v,\ v \in V))$ is
a tree realization of $\cC$ that extends $(T,\omega)$, then 
for all $v \in V$, $\dim(C_v^*) \leq \dim(C_v)$.
\end{quote}
Consequently, we have that $\k(\cC;T,\omega) = \k(\cM(\cC;T,\omega))$,
$\k^+(\cC;T,\omega) = \k^+(\cM(\cC;T,\omega))$, and 
$\k^\tot(\cC;T,\omega) = \k^\tot(\cM(\cC;T,\omega))$.
The fact that $\cM(\cC;T,\omega)$ minimizes local constraint code dimension
at each vertex of $T$ will be central to the 
derivation of our results in the sections to follow.

We will henceforth consistently use the notation $\cS_e^*$ and $C_v^*$
to denote state spaces and local constraint codes in a minimal tree
realization $\cM(\cC;T,\omega)$. Exact expressions for the dimensions 
of $\cS_e^*$ and $C_v^*$ in $\cM(\cC;T,\omega)$ are known 
\cite{For01},\cite{For03}. Recall that for an edge $e$ of $T$, we denote by
$T^{(e)}$ and $\oT^{(e)}$ the two components of $T \del e$. 
Let us further define $J(e) = \omega^{-1}(V(T^{(e)}))$ and 
$\oJ(e) = \omega^{-1}(V(\oT^{(e)}))$. We then have 
\beq
\dim(\cS_e^*) = \dim(\cC) - \dim(\cC_{J(e)}) - \dim(\cC_{\bar{J}(e)}).
\label{dimSe*}
\eeq
%
Next, consider any vertex $v$ in $T$. If $v$ has degree $\d$, then 
$T - v$ has components $T_i^{(v)}$, $i = 1,2,\ldots,\d$. Define
$J_i =  \omega^{-1}(V(T_i^{(v)}))$, for $i = 1,2,\ldots,\d$.
Then \cite[Theorem 1]{For03},
\beq
\dim(C_v^*) = \dim(\cC) - \sum_{i=1}^\d \dim(\cC_{J_i}).
\label{dimCv*}
\eeq

In summary, the minimal tree realization $\cM(\cC;T,\omega)$ 
is an exact solution to the problem of determining the 
minimum-complexity extension of a tree decomposition $(T,\omega)$ 
of a code $\cC$. Moreover, $\cM(\cC;T,\omega)$ minimizes,
among realizations in $\mfR(\cC;T,\omega)$, 
any reasonable measure of complexity, be it state-space complexity 
or constraint complexity. Unfortunately, when we move to 
realizations on graphs with cycles, there appear to be 
no ``canonical'' minimal realizations with properties 
similar to those of minimal tree realizations. In fact, 
if $(\cG,\omega)$ is a graph decomposition of a code $\cC$,
where $\cG$ is a graph with cycles, there need not even be a realization $\G$ 
that simultaneously achieves $\min_{\G \in \mfR(\cC;\cG,\omega)} \k(\G)$
and $\min_{\G \in \mfR(\cC;\cG,\omega)} \k^+(\G)$. An example of
such a graph decomposition is given in Appendix~\ref{example_app}.


Thus, given a code $\cC$ and a graph $\cG$ containing cycles, 
the problem of finding realizations of $\cC$ on $\cG$ with the
least possible $\k$-complexity, $\k^+$-complexity, or 
$\k^\tot$-complexity
(within some interesting sub-class of realizations of $\cC$ on $\cG$) 
is much harder to solve than the corresponding problem for cycle-free graphs.
In the next section, we present a simple but valuable tool that will enable
us to derive non-trivial lower bounds on the constraint
complexity of realizations of a code on an arbitrary graph. These bounds 
could be used, for example, to determine whether or not the complexity
of a given realization is close to the least possible.

\section{Cut-Set Bounds\label{cutset_bnd_section}}

Let $\cG = (V,E)$ be a connected graph. A partition 
$(V',V'')$ of $V$ is said to be \emph{separated} by an 
edge cut $X \subseteq E$ if, for each pair of vertices
$v' \in V'$ and $v'' \in V''$, any path in $\cG$ that 
joins $v'$ to $v''$ passes through some edge $e \in X$.
We remark that if $\cG \setminus X$ has more than two 
components, then there is more than one partition of $V$ that
is separated by $X$.

The Edge-Cut Bound, stated below, is a result of
fundamental importance in the study of graphical realizations. 
This result was originally observed by Wiberg, Loeliger and Koetter
\cite{wiberg},\cite{WLK95}, but the version we give here is
due to Forney \cite[Corollary~4.4]{For01}.

\begin{theorem}[The Edge-Cut Bound]
Let $\G = (\cG,\omega,(C_v,v\in V),(\cS_e,e\in E))$ be 
a realization of a code $\cC$ on a connected graph $\cG = (V,E)$. 
If $(V',V'')$ is a partition of $V$ separated by an edge cut $X \subseteq E$,
then, defining $J' = \omega^{-1}(V')$ and $J'' = \omega^{-1}(V'')$,
we have
$$
\sum_{e \in X}\dim(\cS_e) \geq \dim(\cC) - \dim(\cC_{J'}) - \dim(\cC_{J''}).
$$
\label{edge-cut_bound}
\end{theorem}

The edge-cut bound can be used to derive useful lower bounds on the 
state-space complexity of a graphical realization; see, for example, 
\cite{CFV99}. To deal with constraint complexity, however, we will need a 
closely-related bound that uses vertex cuts instead of edge cuts.

We introduce here some terminology that we will use to state our
vertex-cut bound. For $v \in V$, let $N(v) = \{u \in V: \{u,v\} \in E\}$ 
denote the set of neighbours of $v$ in $\cG$. Furthermore, 
for $W \subseteq V$, define $N(W) = \bigcup_{v \in W} N(v)$.

\begin{definition}
An ordered collection $(V_0,V_1,\ldots,V_\d)$, $\d \geq 0$, of subsets
of $V$ is said to be a \emph{star partition} of $V$, if the $V_i$'s form 
a partition of $V$ (\emph{i.e.}, the $V_i$'s are pairwise disjoint,
and $\bigcup_{i=0}^\d V_i = V$), and for each $i \in \{1,2,\ldots,\d\}$, 
we have $N(V_i) \subseteq V_i \cup V_0$.
\label{vcut_property_def}
\end{definition}

The definition has been worded so as to allow some of the 
$V_i$'s to be empty sets. When $V_i$ is non-empty for at most one 
$i \geq 1$, a star partition is simply a partition. 
When at least two $V_i$'s other than $V_0$ are non-empty, then 
a star partition is a partition that arises from a vertex cut of $\cG$,
as we now explain. For any $i > j \geq 1$, if $V_i$ and $V_j$ 
are both non-empty, then the above definition simply says that 
any path between a vertex in $V_i$ and a vertex in $V_j$ 
must pass through $V_0$. Thus, if at least two $V_i$'s other
than $V_0$ are non-empty, then $V_0$ is a vertex cut of $\cG$.
Conversely, if $V_0$ is a vertex cut of $\cG$, and 
$\cG_1,\cG_2,\ldots,\cG_\delta$, $\d \geq 2$, are the 
(non-empty) components of $\cG - V_0$, then, 
setting $V_i = V(\cG_i)$ for $i = 1,2,\ldots,\d$, we see
that $(V_0,V_1,\ldots,V_\d)$ is a star partition of $V$.
The graph on the left in Figure~\ref{vcut_bnd_fig}, which 
depicts a typical situation covered by the definition, should 
also explain the nomenclature.

\begin{theorem}[The Vertex-Cut Bound]
Let $\G = (\cG,\omega,(C_v,v\in V),(\cS_e,e\in E))$ be 
a realization of a code $\cC$ on a connected graph $\cG = (V,E)$. 
If $(V_0,V_1,\ldots,V_\d)$ is a star partition of $V$, then, 
defining $J_i = \omega^{-1}(V_i)$ for $i = 1,2,\ldots,\d$, we have
$$
\sum_{v \in V_0}\dim(C_v) \geq \dim(\cC) - \sum_{i=1}^\d \dim(\cC_{J_i}).
$$
\label{vertex-cut_bound}
\end{theorem}

\begin{figure}
\epsfig{file=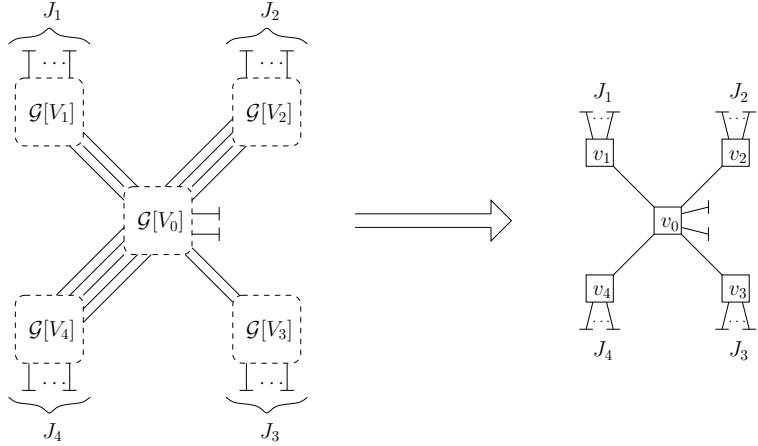, width=10cm}
\caption{A depiction of the construction in the proof of 
the Vertex-Cut Bound. $\cG[V_i]$ denotes the subgraph of $\cG$ 
induced by the vertices in $V_i$.}
\label{vcut_bnd_fig}
\end{figure}

\begin{proof}
Let $\mfB$ denote the full behaviour of $\G$. 
If $\b = ((x_i,i \in I), (\s_e, e \in E))$
is a global configuration in $\mfB$, then given an $X \subseteq E$, 
we will use ${\b|}_X$ to denote the projection $(\s_e, e \in X)$. 
We further set ${\mfB|}_X = \{{\b|}_X: \b \in \mfB\}$.

For $i = 1,2,\ldots,\d$, let $X_i$ be the set of edges of $\cG$ with
exactly one end-point in $V_i$, so that the other end-point is
necessarily in $V_0$. We then define 
$$
{\mfB|}_{V_i} = \{({\b|}_{J_i}, {\b|}_{X_i}):\ \b \in \mfB\},
$$
for $i = 1,2,\ldots,\d$. Furthermore, set $J_0 = \omega^{-1}(V_0)$
and $X_0 = \bigcup_{i=1}^\d X_i$, and define
$$
{\mfB|}_{V_0} = \{({\b|}_{J_0}, {\b|}_{X_0}):\ \b \in \mfB\}.
$$

Now, consider the ``star-shaped'' tree $T = (V_T,E_T)$ 
consisting of a single internal vertex $v_0$ of degree $\d$,
whose neighbours $v_1,v_2,\ldots,v_\d$ are all the leaves of $T$. 
Thus, $V_T = \{v_0,v_1,\ldots,v_\d\}$ and 
$E_T = \{\{v_0,v_i\}:\ i = 1,2,\ldots,\d\}$.
Define the index mapping
$\alpha: I \to V_T$ as follows: $\alpha(j) = v_i$ iff $\omega(j) \in V_i$.
Note that, for $i = 0,1,2,\ldots,\d$, we have 
$\alpha^{-1}(V_i) = \omega^{-1}(V_i) = J_i$. The construction of 
the tree decomposition $(T,\alpha)$ from $(\cG,\omega)$ is depicted
in Figure~\ref{vcut_bnd_fig}.

We next extend the tree decomposition $(T,\alpha)$ to a tree model 
$\widehat{\G} = ((T,\alpha,(\widehat{\cS}_e, e \in E_T),
(\widehat{C}_v, v \in V_T))$ by setting 
$\widehat{C}_{v_i} = {\mfB|}_{V_i}$ for $i = 0,1,2,\ldots,\d$,
and $\widehat{\cS}_{\{v_0,v_i\}} = {\mfB|}_{X_i}$ for $i = 1,2,\ldots,\d$.
From the fact that $\G$ is a realization of $\cC$, it readily follows
that $\widehat{\G}$ is a tree realization of $\cC$.

Recalling that the minimal tree realization $\cM(\cC;T,\alpha)$
minimizes the local constraint code dimension at each vertex of $T$,
we obtain via (\ref{dimCv*}), 
$$
\dim(\widehat{C}_{v_0}) \geq \dim(\cC) - \sum_{i=1}^\d \dim(\cC_{J_i}).
$$
We complete the proof by observing that 
$$
\dim(\widehat{C}_{v_0}) = \dim({\mfB|}_{V_0}) 
\leq \sum_{v \in V_0} \dim({\mfB|}_v) = \sum_{v \in V_0} \dim(C_v).
$$
\end{proof}

The following useful corollary is an immediate consequence of 
the Vertex-Cut Bound.

\begin{corollary}
Let $(\cG,\omega)$ be a graph decomposition of a code $\cC$,
where $\cG$ is a connected graph. 
For a vertex cut $W$ of $\cG$, if $\cG_1,\cG_2,\ldots,\cG_\d$ 
are the components of $\cG - W$, then define 
$\lambda(W) = \dim(\cC) - \sum_{i=1}^\d \dim(\cC_{J_i})$,
where $J_i = \omega^{-1}(V(\cG_i))$ for $i = 1,2,\ldots,\d$.
Then, for any realization $\G = (\cG,\omega,(C_v,v\in V),(\cS_e,e\in E))$ 
of $\cC$ that extends $(\cG,\omega)$, we have
$$
\sum_{v \in W}\dim(C_v) \geq \lambda(W).
$$
\label{vcut_cor}
\end{corollary}

In its most straightforward application, the Vertex-Cut Bound, 
via the above corollary, can be used in conjunction with constrained
optimization techniques to find lower bounds on $\k(\cC;\cG,\omega)$,
$\k^+(\cC;\cG,\omega)$ and $\k^\tot(\cC;\cG,\omega)$.
Indeed, if $(\cG,\omega)$ is a graph decomposition of $\cC$, and  
$W_1,W_2,\ldots,W_t$ are vertex cuts of $\cG$, then, by
Corollary~\ref{vcut_cor}, the dimensions of the local constraint codes 
$C_v$, $v \in V$, in any $\G \in \mfR(\cC;\cG,\omega)$ must satisfy 
$\sum_{v \in W_i}\dim(C_v) \geq \lambda(W_i)$, $i = 1,2,\ldots,t$.
Thus, for example, $\k^+(\cC;\cG,\omega)$ is lower bounded by the
solution to the following linear programming problem in the variables
$\xi_v$, $v \in V$: given a collection of vertex cuts $W_1,W_2,\ldots,W_t$
of $\cG$,
$$
\text{minimize}\ \ \sum_{v \in V} \xi_v,
\ \ \text{subject to}\ \ 
\sum_{w \in W_i} \xi_w \geq \l(W_i),\ i = 1,2,\ldots, t.
$$

However, we do not pursue this angle any further in this paper. 
Instead, we will henceforth restrict our attention to the 
$\k$-complexity measure, for which we will derive a suite
of lower bounds, again based on the Vertex-Cut Bound, which 
unearth some interesting connections with graph theory, 
and moreover, are amenable to further mathematical analysis. 
The bounds we derive rely on the notion of vertex-cut trees 
introduced in the next section.

\section{Vertex-Cut Trees\label{vctree_section}}

We begin with a simple lemma, which plays a role in our
definition of a vertex-cut tree below.

\begin{lemma}
Suppose that $V_0, V_1, \ldots, V_\d$ are subsets of $V$ such that
$\bigcup_{i=0}^\d V_i = V$. If, for each pair of distinct indices 
$i,j$, we have $V_i \cap V_j \subseteq V_0$, then 
$(V_0,V_1 \setminus V_0, \ldots, V_\d \setminus V_0)$ is a partition of $V$.
\label{cover_lemma}
\end{lemma}
\begin{proof}
It is evident that $V_0 \cup \bigcup_{i=1}^\d (V_i \setminus V_0) = 
\bigcup_{i=0}^\d V_i = V$. If $i,j > 0$, $i \neq j$, then 
$(V_i \setminus V_0) \cap (V_j \setminus V_0) = (V_i \cap V_j) \setminus V_0
= \emptyset$, since $V_i \cap V_j \subseteq V_0$.
\end{proof}

For the main definition of this section, we introduce some convenient 
notation, which will henceforth be used consistently.
If $A$ and $B$ are sets, and $f: A \to 2^B$ is a mapping from $A$ 
to the power set of $B$, then for any $X \subseteq A$, 
we define $f(X) = \bigcup_{x \in X} f(x)$. Also, recall
that if $z$ is a vertex of degree $\d$ in a tree $T$, then
the components of $T-z$ are denoted by $T_i^{(z)}$, $i = 1,2,\ldots,\d$.

\begin{definition}
Let $\cG$ be a connected graph. A \emph{vertex-cut tree} of $\cG$
is a data structure $(T,\beta)$, where $T$ is a tree, 
and $\beta: V(T) \to 2^{V(\cG)}$ is a mapping with the following properties:
\begin{itemize}
\item[(VC1)] $\beta(V(T)) = V(\cG)$;
\item[(VC2)] for each pair $x,y \in V(T)$, if $z \in V(T)$ is any vertex 
that lies on the unique path between $x$ and $y$ in $T$, 
then $\beta(x) \cap \beta(y) \subseteq \beta(z)$;
\item[(VC3)] for each $z \in V(T)$, $(\beta(z),V_1,V_2,\ldots,V_\d)$
is a star partition of $V(\cG)$, where $\d$ is the degree of $z$, and 
$V_i = \beta(V(T_i^{(z)})) \setminus \beta(z)$ for $i = 1,2,\ldots,\d$.
\end{itemize}
A vertex-cut tree $(T,\beta)$ in which $T$ is a path is
called a \emph{vertex-cut path}.
\label{vcut_tree_def}
\end{definition}

Note that, by Lemma~\ref{cover_lemma}, conditions (VC1) and (VC2) in the
above definition imply that, for each $z \in V(T)$, with $\d$ and $V_i$,
$i = 1,2,\ldots,\d$, as in condition (VC3), 
$(\beta(z),V_1,V_2,\ldots,V_\d)$ is a partition of $V(\cG)$. 
Thus, for $(T,\beta)$ satisfying conditions (VC1) and (VC2), 
condition (VC3) is met iff, for each $z \in V(T)$, we have
$N(V_i) \subseteq V_i \cup \beta(z)$ for all $i \geq 1$.

Trivial vertex-cut trees (and paths) always exist for a graph $\cG$ ---
given any tree $T$, pick a vertex $z_0 \in V(T)$, and define 
a mapping $\beta$ by setting $\beta(z_0) = V(\cG)$, and 
$\beta(z) = \emptyset$ for $z \neq z_0$. This allows us to make the
following definition.

\begin{definition}
Let $\cG$ be a connected graph. 
The \emph{vc-width} of a vertex-cut tree $(T,\beta)$ of $\cG$ is defined as 
$\max_{z \in V(T)} |\beta(z)|$, and is denoted by $\text{vc-width}(T,\beta)$.
The \emph{vc-treewidth} (resp.\ \emph{vc-pathwidth}) of $\cG$ is the 
least vc-width among all vertex-cut trees (resp.\ vertex-cut paths) of $\cG$,
and is denoted by $\k_{\mbox{\scriptsize\emph{vc-tree}}}(\cG)$ 
(resp.\ $\k_{\mbox{\scriptsize\emph{vc-path}}}(\cG)$).
\label{vcwidth_def}
\end{definition}

Thus, for any graph $\cG$, we have 
$0 < \k_\vctree(\cG) \leq \k_\vcpath(\cG) \leq |V(\cG)|$.
The vc-treewidth of any tree is equal to one. Indeed, if $T$ is a tree,
then $(T,\beta)$, defined by $\beta(z) = \{z\}$ for all $z \in V(T)$,
is a vertex-cut tree of $T$, with vc-width equal to one.

\begin{example}
Let $\cG$ be an $n$-cycle with vertices $v_0,v_1,\ldots,v_{n-1}$,
labeled in cyclic order. Let $P$ be a path with $n-1$ vertices, 
which in the linear order defined by the path, are 
labeled $z_1,z_2,\ldots,z_{n-1}$. To be precise, $z_0$ is one of the two
leaves, and for $i=1,2,\ldots,n-1$, $z_i$ is adjacent to $z_{i-1}$ in $P$.
If we define the mapping $\beta:V(P) \to 2^{V(\cG)}$ as
$\beta(z_i) = \{v_0,v_i\}$ for $i = 1,2,\ldots,n-1$, then 
$(P,\beta)$ is a vertex-cut path of $\cG$, of vc-width two. 
It is not difficult to verify that $\cG$ has no vertex-cut tree 
of vc-width one, and hence, 
$\k_{\mbox{\scriptsize\emph{vc-tree}}}(\cG) = 
\k_{\mbox{\scriptsize\emph{vc-path}}}(\cG) = 2$.
\label{cycle_example}
\end{example}

Our definition of vertex-cut trees may appear at first to be an 
artificial construct brought in solely for the purpose of 
finding applications for the Vertex-Cut Bound. However, this is
far from being the case. Vertex-cut trees are very closely related
to junction trees, commonly associated with belief propagation algorithms
in Bayesian networks \cite{jensen} and coding theory \cite{AM00}.
Another close relative of vertex-cut trees is the data structure 
known as tree decomposition (of a graph)
\cite{RS-I},\cite{arnborg},\cite{bod93},
which has received considerable attention in the graph theory and
computer science literatures.

\begin{definition}
A \emph{tree decomposition} of a connected graph $\cG$
is a data structure $(T,\beta)$, where $T$ is a tree, 
and $\beta: V(T) \to 2^{V(\cG)}$ is a mapping with the following properties:
\begin{itemize}
\item[(T1)] $\beta(V(T)) = V(\cG)$;
\item[(T2)] for each pair $x,y \in V(T)$, if $z \in V(T)$ is any vertex 
that lies on the unique path between $x$ and $y$ in $T$, 
then $\beta(x) \cap \beta(y) \subseteq \beta(z)$;
\item[(T3)] for each pair of adjacent vertices $u,v \in V(\cG)$, 
there exists a $z \in V(T)$ such that $\{u,v\} \subseteq \beta(z)$.
\end{itemize}
A tree decomposition $(T,\beta)$ in which $T$ is a path is
called a \emph{path decomposition}.
\label{tree_decomp_def}
\end{definition}

The \emph{width} of a tree decomposition as above is defined to be
$\max_{z \in V(T)} |\beta(z)| - 1$. The \emph{treewidth} 
(resp.\ \emph{pathwidth}) of a graph $\cG$, denoted by $\k_\tree(\cG)$
(resp.\ $\k_\path(\cG)$), is the minimum among the widths of all its 
tree (resp.\ path) decompositions. Note that if $\cG$ has at least one 
non-loop edge, then, because of (T3), any tree decomposition of $\cG$
must have width at least one. Thus, for any such graph $\cG$, we 
have $0 < \k_\tree(\cG) \leq \k_\path(\cG) \leq |V(\cG)|-1$.
It is not hard to check that any tree with at least two vertices
has treewidth equal to one\footnote{Without the `$-1$' in the 
definition of width of a tree decomposition, a tree with at 
least two vertices would have treewidth equal to two.},
and that if $\cG$ is an $n$-cycle,
then $\k_\tree(\cG) = \k_\path(\cG) = 2$.

\begin{lemma}
Any tree decomposition of a connected graph $\cG$ is also a vertex-cut 
tree of $\cG$. Hence, $\k_{\mbox\scriptsize\emph{vc-tree}}(\cG) 
\leq \k_{\mbox\scriptsize\emph{tree}}(\cG) + 1$, 
and $\k_{\mbox\scriptsize\emph{vc-path}}(\cG) 
\leq \k_{\mbox\scriptsize\emph{path}}(\cG) + 1$.
\label{tree_vctree_lemma}
\end{lemma}
\begin{proof}
Let $(T,\beta)$ be a tree decomposition of $\cG$. We only have to show
that $(T,\beta)$ satisfies condition (VC3) of Definition~\ref{vcut_tree_def}.
Consider any $z \in V(T)$, with degree $\d$, and let $V_i$, 
$i=1,\ldots,\d$, be defined as in condition (VC3). By virtue of 
(T1), (T2) and Lemma~\ref{cover_lemma}, $(\beta(z),V_1,\ldots,V_\d)$
is a partition of $V(\cG)$. Thus, we must show that
$N(V_i) \subseteq V_i \cup \beta(z)$ for $i = 1,\ldots,\d$.

Suppose, to the contrary, that there exists a $v \in V_i$ 
such that for some $u \in N(v)$, we have $u \notin V_i \cup \beta(z)$. 
Thus, by definition of $V_i$, we have 
$u \notin \beta(V(T_i^{(z)})) \cup \beta(z)$,
while $v \in \beta(x) \setminus \beta(z)$ 
for some $x \in V(T_i^{(z)})$. Now, by condition (T3), we have
$\{u,v\} \subseteq \beta(y)$ for some $y \in V(T)$.
In particular, $v \in \beta(x) \cap \beta(y)$.
On the other hand, by our assumption on $u$, we must have $y \neq z$ and 
$y \notin V(T_i^{(z)})$. Thus, $y \in V(T_j^{(z)})$ for some $j \neq i$. 
This means that the vertex $z$ lies on the unique
path in $T$ joining $x$ and $y$, and hence, by (T2), $v \in \beta(z)$,
which contradicts the existence of $v$ as postulated.
\end{proof}

\begin{figure}
\epsfig{file=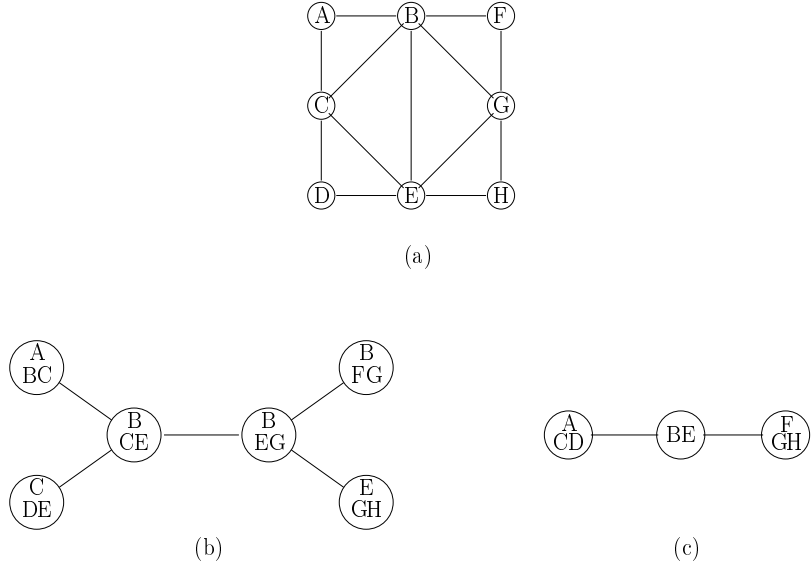, height=7.5cm}
\caption{(a)\ A graph $\cG$ on eight vertices. (b)\ A tree decomposition
of $\cG$ with width 2, and vc-width 3. 
(c) A vertex-cut path of $\cG$ of vc-width 3
that is not a path decomposition of $\cG$.}
\label{vcut_trees}
\end{figure}

A graph $\cG$ can have a vertex-cut tree that is not a tree decomposition.
Figure~\ref{vcut_trees} shows an example of a vertex-cut path 
of a graph $\cG$ that is not a path decomposition of $\cG$.
Also, the inequality $\k_\vctree(\cG) \leq \k_\tree(\cG) + 1$ 
in Lemma~\ref{tree_vctree_lemma} can hold with equality ---
for the graph $\cG$ in Figure~\ref{vcut_trees}(a), it is possible 
to show that $\k_\tree(\cG) = 2$, while 
$\k_\vctree(\cG) = \k_\vcpath(\cG) = 3$.

\section{Lower Bounds on $\k$-Complexity\label{lobnds_section}}

Vertex-cut trees allow us to derive lower bounds on 
the $\k$-complexity of a graphical realization of a code, 
as we now show. Let $(\cG,\omega)$ be a graph decomposition of a 
code $\cC$, and let $(T,\beta)$ be a vertex-cut tree of $\cG$. 
For each $z \in V(T)$, set $J_i = \omega^{-1}(V_i)$
for $i = 1,2,\ldots,\d$, where the $V_i$'s are as defined in 
condition~(VC3) of Definition~\ref{vcut_tree_def}. Further define
\beq
m(z) = \dim(\cC) - \sum_{i=1}^\d \dim(\cC_{J_i}).
\label{m_def}
\eeq
Now, consider any graphical realization, 
$\G = (\cG,\omega,(\cS_e,e \in E),(C_v, v\in V))$,
of $\cC$ that extends $(\cG,\omega)$. By the Vertex-Cut Bound, we have,
for each $z \in V(T)$, $\sum_{v \in \beta(z)} \dim(C_v) \geq m(z)$.
Since $\sum_{v \in \beta(z)} \dim(C_v) \leq 
\left(\max_{v \in V(\cG)} \dim(C_v)\right) \, |\beta(z)| 
= \k(\G) \, |\beta(z)|$, we obtain
$$
\k(\G) \, |\beta(z)| \geq m(z).
$$
Maximizing over all $z \in V(T)$, we get
\beq
\k(\G) \cdot \text{vc-width}(T,\beta) 
\geq \max_{z \in V(T)} m(z)
\ \define \ \mu(\cC;\omega,\beta)
\label{mu_def}
\eeq
We thus have the following proposition.
\begin{proposition}
Let $(\cG,\omega)$ be a graph decomposition of a code $\cC$,
and let $(T,\beta)$ be a vertex-cut tree of $\cG$. 
Then, for any graphical realization 
$\G \in \mfR(\cC;\cG,\omega)$, we have
$$
\k(\G) \geq \frac{\mu(\cC;\omega,\beta)}
{\emph{vc-width}(T,\beta)}.
$$
\label{k_Gamma_lobnd}
\end{proposition}

The lower bound in the above proposition can be brought into a form more
convenient for further analysis. To do this, we will 
construct a tree decomposition $(T,\gm)$ of $\cC$ such that 
$\mu(\cC;\omega,\beta) = \k(\cC;T,\gm)$. 
So, once again, let $(\cG,\omega)$ be a given graph decomposition 
of $\cC$, and $(T,\beta)$ a given vertex-cut tree of $\cG$. We will first
give a recipe for constructing tree decompositions of $\cC$ from 
$(\cG,\omega)$ and $(T,\beta)$. Then, we will show how the main ingredient
of the recipe may be chosen so that the resulting tree decomposition
$(T,\gm)$ satisfies $\mu(\cC;\omega,\beta) = \k(\cC;T,\gm)$.

\begin{figure}
\epsfig{file=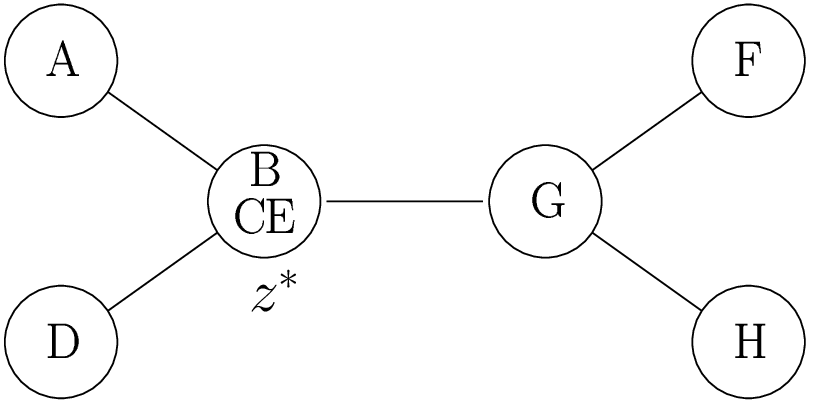, width=5.2cm}
\caption{The mapping $\alpha:V(T) \to 2^{V(\cG)}$, with
the vertex $z^*$ chosen as shown, constructed
from the vertex-cut tree in Figure~\ref{vcut_trees}(b).}
\label{alpha_fig}
\end{figure}

For any pair of vertices $x,y \in V(T)$, let $(x,y]$ denote the set of
vertices on the unique path between $x$ and $y$ in $T$, including
$y$, but not including $x$. Also, we will use $\beta(x,y]$,
instead of the more cumbersome $\beta((x,y])$, to denote the
set $\bigcup_{z \in (x,y]} \beta(z)$.
Pick an arbitrary vertex $z^*$ in $V(T)$. Define the mapping 
$\a:V(T) \to 2^{V(\cG)}$ as follows: $\a(z^*) = \beta(z^*)$, 
and for $z \neq z^*$,
$$
\a(z) = \beta(z) \setminus \beta(z,z^*].
$$
An example of such a mapping $\alpha$ constructed from
the vertex-cut tree $(T,\beta)$ in Figure~\ref{vcut_trees}(b) 
is depicted in Figure~\ref{alpha_fig}.

We record in the next two lemmas some properties of the mapping $\a$ that 
we use in the sequel. Recall our convention that for $X \subseteq V(T)$,
$\a(X) = \bigcup_{x \in X} \a(x)$.

\begin{lemma}
For $z \in V(T)$, if $T^{(z)}$ is any component of $T - z$, then
$$
\a(V(T^{(z)})) = 
\begin{cases}
\beta(V(T^{(z)})) & \text{ if $z^* \in V(T^{(z)})$} \\
\beta(V(T^{(z)})) \setminus \beta(z) & \text{ if $z^* \notin V(T^{(z)})$}.
\end{cases}
$$
\label{alpha_lemma}
\end{lemma}
\begin{proof}
Suppose first that $z^* \in V(T^{(z)})$, and set 
$X = V(T^{(z)}) \setminus \{z^*\}$. It then follows from 
the definition of $\a$ that 
$$
\bigcup_{x \in X} \a(x) 
= \left(\bigcup_{x \in X} \beta(x)\right) \setminus \beta(z^*).
$$
Therefore, $\a(V(T^{(z)})) = \a(X) \cup \a(z^*) = 
(\beta(X) \setminus \beta(z^*)) \cup \beta(z^*) = \beta(V(T^{(z)}))$.

\begin{figure}
\epsfig{file=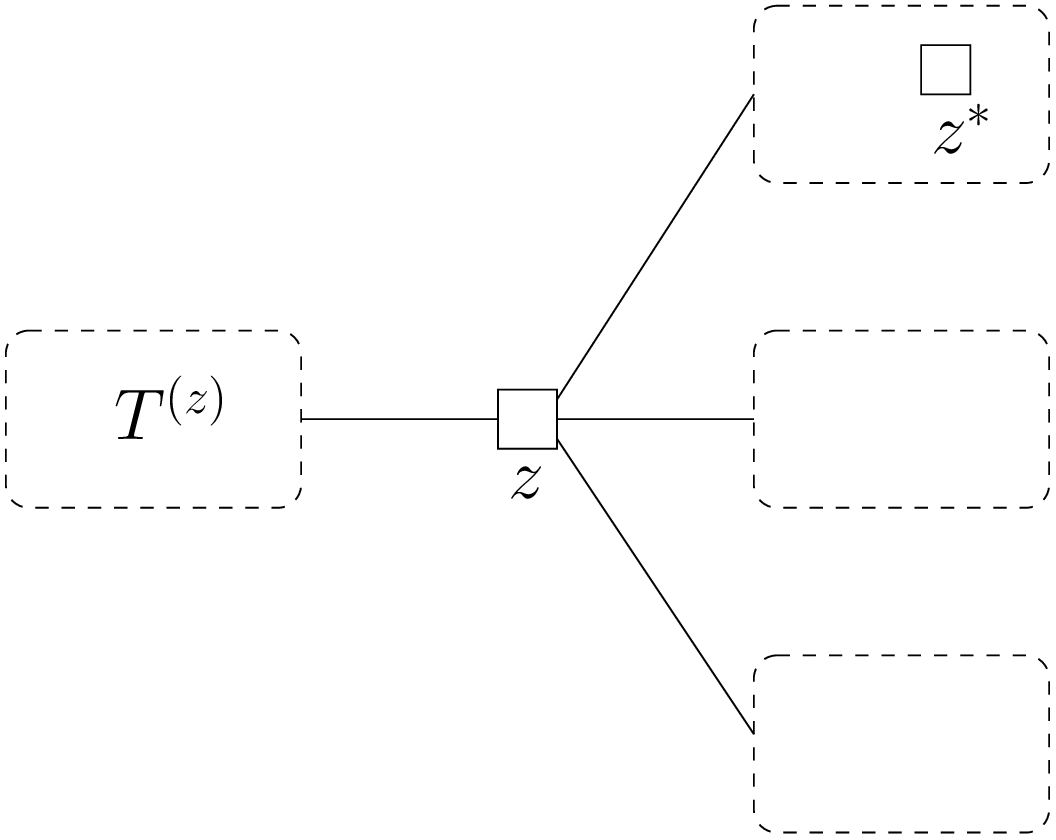, width=6.5cm}
\caption{A depiction of the situation when $z^* \notin V(T^{(z)})$.
Dashed ovals represent components of $T-z$.}
\label{z*_notin_fig}
\end{figure}

Now, consider the case when $z^* \notin V(T^{(z)})$, as depicted in
Figure~\ref{z*_notin_fig}. In this case, the definition of $\a$
implies that 
$$
\bigcup_{x \in V(T^{(z)})} \a(x) 
= \left(\bigcup_{x \in V(T^{(z)})} \beta(x)\right) \setminus \beta[z,z^*],
$$
where $[z,z^*]$ denotes the set of vertices on the path between
$z$ and $z^*$ in $T$, including both $z$ and $z^*$. Note that,
as a consequence of condition (VC2) of Definition~\ref{vcut_tree_def}, 
we have, for any $x \in V(T^{(z)})$, 
$\beta(x) \cap \beta[z,z^*] = \beta(x) \cap \beta(z)$. Hence,
$$
\left(\bigcup_{x \in V(T^{(z)})} \beta(x)\right) \setminus \beta[z,z^*]
= \left(\bigcup_{x \in V(T^{(z)})} \beta(x)\right) \setminus \beta(z),
$$
which proves the desired result.
\end{proof}

\begin{lemma}
The sets $\a(z)$, $z \in V(T)$, form a partition of $V(\cG)$.
\label{alpha_partition}
\end{lemma}
\begin{proof}
Let $T_i^{(z^*)}$, $i = 1,2,\ldots,\d$, denote the components
of $T - z^*$. Observe that, by Lemma~\ref{alpha_lemma},
$$
\bigcup_{z \neq z^*} \a(z) = \bigcup_{i=1}^\d \a(V(T_i^{(z^*)}))
= \bigcup_{i=1}^\d \beta(V(T_i^{(z^*)})) \setminus \beta(z^*) 
= \left(\bigcup_{z \neq z^*} \beta(z)\right) \setminus \beta(z^*).
$$
Therefore, 
$\bigcup_{z \in V(T)} \a(z) = \bigcup_{z \in V(T)} \beta(z) = V(\cG)$,
by condition (VC1) of Definition~\ref{vcut_tree_def}.

Next, we want to show that $\a(z) \cap \a(z') = \emptyset$ for $z \neq z'$.
This is true by definition of $\a$ if either $z = z^*$ or $z' = z^*$. 
So, we henceforth assume that $z$ and $z'$ are distinct vertices 
in $V(T) \setminus \{z^*\}$.

Suppose that there exists a $v \in \a(z) \cap \a(z')$. We then have
\begin{itemize}
\item[(1)] $v \in \beta(z)$, but $v \notin \beta(y)$ for any $y \in (z,z^*]$; 
and
\item[(2)] $v \in \beta(z')$, but $v \notin \beta(y)$ for any $y \in (z',z^*]$.
\end{itemize}
In particular, we see that $v \notin \beta(z^*)$. It also follows
from (1) and (2) that $z' \notin (z,z^*]$, and $z \notin (z',z^*]$, 
which together imply that $z^*$ lies on the path between $z$ and $z'$. 
However, by (VC2), this means that $v \in \beta(z^*)$, a contradiction.
\end{proof}

We now construct a tree decomposition $(T,\gm)$ of the index set, $I$,
of $\cC$, by defining an index mapping $\gm:I \to V(T)$ as follows:
for each $i \in I$, set $\gm(i) = z$ if $\omega(i) \in \a(z)$.
In other words, for each $z \in V(T)$, $\gm^{-1}(z) = \omega^{-1}(\a(z))$.
Since the sets $\a(z)$, $z \in V(T)$, form a partition of $V(\cG)$,
the mapping $\gm$ is well-defined. 

\begin{example}
Suppose that $\cC$ is a code of length 10, defined on the index set 
$I = \{1,2,\ldots,10\}$. For the graph $\cG$ shown in 
Figure~\ref{vcut_trees}(a), consider the graph decomposition $(\cG,\omega)$
of $\cC$ defined by $\omega(1) = \omega(2) = A$, 
$\omega(3) = \omega(4) = \omega(5) = B$, $\omega(6) = \omega(7) = E$,
$\omega(8) = F$, and $\omega(9) = \omega(10) = H$. Let $(T,\beta)$
be the vertex-cut tree of $\cG$ in Figure~\ref{vcut_trees}(b), 
from which we obtain the mapping $\a:V(T) \rightarrow 2^{V(\cG)}$ 
depicted in Figure~\ref{alpha_fig}. Based on the last figure, we will 
use the labels $z^*$, $A$, $D$, $G$, $F$ and $H$ to 
identify the vertices of the tree $T$. Then, the index mapping 
$\gm:I \to V(T)$ is given by $\gm(1) = \gm(2) = A$,
$\gm(3) = \gm(4) = \gm(5) = \gm(6) = \gm(7) = z^*$,
$\gm(8) = F$, and $\gm(9) = \gm(10) = H$. 

\label{gamma_example}
\end{example}

Note that our construction of the tree decomposition $(T,\gm)$ 
depends on the choice of the vertex $z^*$, via the mapping $\a$.
Up to this point, our choice of $z^*$ was arbitrary. 
We now specify $z^* \in V(T)$ to be such that 
$m(z^*) = \mu(\cC;\omega,\beta)$. We claim that for the tree decomposition 
$(T,\gm)$ arising from such a choice of $z^*$, we have
$\mu(\cC;\omega,\beta) = \k(\cC;T,\gamma)$.

\begin{proposition}
Given a graph decomposition $(\cG,\omega)$ of a code $\cC$,
and a vertex-cut tree $(T,\beta)$ of $\cG$, there exists a
tree decomposition $(T,\gamma)$ of $\cC$ such that 
$\mu(\cC;\omega,\beta) = \k(\cC;T,\gamma)$.
\label{mu_k_prop}
\end{proposition}
\begin{proof}
Pick a $z^* \in V(T)$ such that $m(z^*) = \mu(\cC;\omega,\beta)$,
and construct the tree decomposition $(T,\gm)$ as described above.
Note that, by (\ref{dimCv*}),
$$
\k(\cC;T,\gm) = \k(\cM(\cC;T,\gm)) = \max_{z \in V(T)} k(z),
$$
where, for a vertex $z \in V(T)$ of degree $\d$, 
$k(z) \define \dim(\cC) - \sum_{i=1}^\d \dim(\cC_{K_i})$,
with $K_i = \gm^{-1}\left(V(T_i^{(z)})\right)$ for $i = 1,2,\ldots,\d$.
Recall from the definition of $\gm$ that 
$\gm^{-1}(z) = \omega^{-1}(\a(z))$ for any $z \in V(T)$. 
Hence, $K_i = \omega^{-1}\left(\a(V(T_i^{(z)}))\right)$, $i = 1,2,\ldots,\d$.

To prove the proposition, we show that $k(z) \leq m(z)$ for all 
$z \in V(T)$, with equality holding if $z = z^*$. From (\ref{m_def}), we 
see that $m(z) = \dim(\cC) - \sum_{i=1}^\d \dim(\cC_{J_i})$, where,
for  $i = 1,2,\ldots,\d$. 
$J_i = \omega^{-1}\left(\beta(V(T_i^{(z)})) \setminus \beta(z)\right)$.
We must therefore show that, for each $z \in V(T)$, we have
$J_i \subseteq K_i$ for all $i$, which would prove that $k(z) \leq m(z)$;
and furthermore, when $z = z^*$, we have $J_i = K_i$ for all $i$,
which would prove that $k(z^*) = m(z^*)$.

So, consider any $z \in V(T)$. From Lemma~\ref{alpha_lemma}, we see
that if $T_i^{(z)}$ is any component of $T-z$, then 
\beq
\beta(V(T_i^{(z)})) \setminus \beta(z) \subseteq \a(V(T_i^{(z)}))
\label{J_subset_K}
\eeq
Hence, $J_i \subseteq K_i$. Moreover, when $z = z^*$, we always have
$z^* \notin V(T_i^{(z)})$, and so, again by Lemma~\ref{alpha_lemma}, 
equality holds in (\ref{J_subset_K}), implying that $J_i = K_i$.
\end{proof}

From Propositions~\ref{k_Gamma_lobnd} and \ref{mu_k_prop}, 
we obtain the following theorem.

\begin{theorem}
Let $(\cG,\omega)$ be a graph decomposition of a code $\cC$,
and let $(T,\beta)$ be a vertex-cut tree of $\cG$. 
Then, there exists a tree decomposition $(T,\gm)$ of $\cC$ such that,
for any graphical realization $\G \in \mfR(\cC;\cG,\omega)$, we have
$$
\k(\G) \geq \frac{\k(\cC;T,\gm)}{\text{\em vc-width}(T,\beta)}.
$$
Hence, $\k(\cC;\cG,\omega) 
\geq \frac{\k(\cC;T,\gm)}{\text{\em vc-width}(T,\beta)}$.
\label{lobnd_thm1}
\end{theorem}

The theorem above is a fundamental result with several important consequences, 
some of which we present here. The first of these is a corollary that 
gives a lower bound on the $\k$-complexity of 
\emph{any} realization of a given code $\cC$ on a graph $\cG$.

\begin{corollary}
For a code $\cC$ and a connected graph $\cG$, we have
$$
\k(\cC;\cG) \geq \frac{\k_{\text{\scriptsize\emph{tree}}}(\cC)}
{\k_{\text{\scriptsize\emph{vc-tree}}}(\cG)}
\geq \frac{\k_{\text{\scriptsize\emph{tree}}}(\cC)}
{\k_{\text{\scriptsize\emph{tree}}}(\cG) + 1}.
$$
\label{cor1}
\end{corollary}
\begin{proof}
Consider an arbitrary index mapping $\omega:I \to V(\cG)$, where $I$ is
the index set of $\cC$. Let $(T,\beta)$ be an \emph{optimal} vertex-cut
tree of $\cG$, by which we mean that $\text{vc-width}(T,\beta) 
= \k_\vctree(\cG)$. By Theorem~\ref{lobnd_thm1}, there exists 
a tree decomposition $(T,\gm)$ of $\cC$ such that 
$$
\k(\cC;\cG,\omega) \geq \frac{\k(\cC;T,\gm)}{\text{vc-width}(T,\beta)}
\geq \frac{\k_\tree(\cC)}{\k_\vctree(\cG)}.
$$
Minimizing over all possible mappings $\omega: I \to V(\cG)$,
we obtain 
$$
\k(\cC;\cG) \geq \frac{\k_\tree(\cC)}{\k_\vctree(\cG)}.
$$
From Lemma~\ref{tree_vctree_lemma}, we also have 
$\k_\vctree(\cG) \leq \k_\tree(\cG) + 1$.
\end{proof}

If, in the above proof, we take $(T,\beta)$ to be an optimal vertex-cut 
path instead, \emph{i.e.}, take $(T,\beta)$ to be a vertex-cut path such
that $\text{vc-width}(T,\beta) = \k_\vcpath(\cG)$, 
then we obtain the next corollary.

\begin{corollary}
For a code $\cC$ and a connected graph $\cG$, we have
$$
\k(\cC;\cG) \geq \frac{\k_{\text{\scriptsize\emph{path}}}(\cC)}
{\k_{\text{\scriptsize\emph{vc-path}}}(\cG)}
\geq \frac{\k_{\text{\scriptsize\emph{path}}}(\cC)}
{\k_{\text{\scriptsize\emph{path}}}(\cG) + 1}.
$$
\label{cor2}
\end{corollary}

The (first) inequality above can be tight, in the sense that
there are examples of codes $\cC$ and graphs $\cG$ for which 
$\k(\cC;\cG) = \lceil\frac{\k_\path(\cC)}{\k_\vcpath(\cG)}\rceil$.

\begin{example}
Take $\cC$ to be the $[24,12,8]$ binary Golay code, 
and let $\mfG$ be the family of graphs consisting of all $n$-cycles.
From Example~\ref{cycle_example}, we know that 
$\k_{\text{\scriptsize\emph{vc-path}}}(\cG) = 2$
for any $\cG \in \mfG$. It is also known that 
$\k_{\text{\scriptsize\emph{path}}}(\cC) = 9$ 
\cite[Example~1]{For03},\cite[Section~5]{vardy}.
Hence, by the bound of Corollary~\ref{cor2},
noting that $\k(\cC;\cG)$ must be an integer, we have 
$\k(\cC;\cG) \geq 5$ for any $\cG \in \mfG$. Thus, $\k(\cC;\mfG) \geq 5$.
The tailbiting trellis realization of the Golay code given in 
\cite{CFV99} has $\k$-complexity equal to 5, 
from which we conclude that $\k(\cC;\mfG) = 5$.
\label{golay_example}
\end{example}

Using Corollary~\ref{cor2} as a starting point, we derive a 
lower bound on the treewidth of an $[n,k,d]$ linear code. 
For this, we will also need the following result 
\cite[Theorem~7.1]{BK96}: if $\cG$ is a graph with treewidth 
at most $\ell$, then the pathwidth of $\cG$ is at most 
$(\ell+1) \log_2 |V(\cG)|$. 

\begin{proposition}
For an $[n,k,d]$ linear code $\cC$, with $n > 1$,
$$
\k_{\text{\scriptsize\emph{tree}}}(\cC) \geq 
\frac{\k_{\text{\scriptsize\emph{path}}}(\cC)}{3 + 2\log_2(n-1)} 
\geq \frac{k(d-1)}{n(3 + 2\log_2(n-1))}.
$$
\label{treewidth_bnd}
\end{proposition}
\begin{proof}
The second inequality above is due to the fact, shown in \cite{LV95},
that\footnote{The result in \cite{LV95} is only explicitly stated 
as a lower bound on the state max-complexity of any conventional 
trellis realization of $\cC$. However, the state max-complexity 
of a graphical realization can never exceed the 
constraint max-complexity of the realization, as noted in
Section~\ref{complexity_section}.}
for an $[n,k,d]$ linear code $\cC$, $\k_\path(\cC) \geq k(d-1)/n$. 

The first inequality is proved as follows. Let $(T,\omega)$ 
be an optimal tree decomposition of $\cC$, so that 
$\k(\cC;T,\omega) = \k_\tree(\cC)$. As mentioned in 
Section~\ref{complexity_section}, $(T,\omega)$ may be chosen 
so that $T$ is a cubic tree, and $\omega$ is a bijection between 
the index set of $\cC$ and the leaves of $T$. Thus, $T$ is a cubic tree with
$n$ leaves, from which it follows that $T$ has $n-2$ internal nodes.
Hence, $|V(T)| = 2n-2$. Since $T$ has treewidth equal to one,
by the result from \cite{BK96} quoted earlier, 
$\k_\path(T) \leq 2 \log_2(2n-2)$.

We now have, via Corollary~\ref{cor2}, 
$$
\k(\cC;T,\omega) \geq \k(\cC;T) \geq \frac{\k_\path(\cC)}{\k_\path(T) + 1}
\geq \frac{\k_\path(\cC)}{2\log_2(2n-2) + 1},
$$
which proves the proposition.
\end{proof}

In particular, Proposition~\ref{treewidth_bnd} shows that for any
linear code $\cC$ of length $n$,
$$
\frac{\k_\path(\cC)}{\k_\tree(\cC)} = O(\log_2 n).
$$
This estimate of the ratio $\frac{\k_\path(\cC)}{\k_\tree(\cC)}$ is
the best possible\footnote{It was conjectured in \cite{K2} that 
for codes $\cC$ of length $n$, $\k_\path(\cC) - \k_\tree(\cC) = O(\log n)$,
but we now do not believe this to be true.},
up to the constant implicit in the $O$-notation.
Indeed, it was shown in \cite{K2} that a sequence of codes $\cC^{(i)}$ 
($i = 1,2,3,\ldots$), with length $n_i = 12(2^i - 1) + 2$, 
$\k_\tree(\cC^{(i)}) = 2$ and $\k_\path(\cC^{(i)}) \geq \frac{1}{2}(i+3)$, 
can be constructed over any finite field $\F$. 
It is clear that $\frac{\k_\path(\cC^{(i)})}{\k_\tree(\cC^{(i)})}$ 
grows logarithmically with codelength $n_i$. 

\medskip

A less precise formulation of the second inequality in 
Proposition~\ref{treewidth_bnd} is also instructive: there
exists a constant $c_0 > 0$ such that, for any $[n,k,d]$
linear code $\cC$, with $n > 1$,
\beq
\k_\tree(\cC) \geq c_0 \frac{k \, d}{n \, \log_2n}.
\label{imprecise_bnd}
\eeq
A noteworthy implication of the above inequality is that
code families of bounded treewidth are not very good from
an error-correcting perspective. Given an integer $t > 0$,
denote by $\TW(t)$ the family of all codes over $\F$
of treewidth at most $t$. A code family $\mfC$ is called 
\emph{asymptotically good} if there exists a sequence of 
$[n_i,k_i,d_i]$ codes $\cC^{(i)} \in \mfC$, with 
$\lim_i n_i = \infty$, such that $\liminf_i k_i/n_i$ and 
$\liminf_i d_i/n_i$ are both strictly positive. The following
result, an easy consequence of (\ref{imprecise_bnd}), resolves 
a conjecture in \cite{K2}.

\begin{corollary}
Let $\cC^{(i)}$, $i=1,2,3$, be any sequence of $[n_i,k_i,d_i]$ codes 
such that 
$$\lim_{i \to \infty} 
\k_{\text{\scriptsize\emph{tree}}}(\cC^{(i)})\frac{\log n_i}{n_i} = 0.$$
Then, either $\lim_{i \to \infty} k_i/n_i = 0$ 
or $\lim_{i\to\infty} d_i/n_i = 0$. In particular, for any $t > 0$, 
the code family $\TW(t)$ is not asymptotically good.
\label{cor3}
\end{corollary}

In fact, a more general result is true. For a fixed integer $\ell > 0$, 
let $\mfG_\ell$ denote the family of all graphs with vc-treewidth 
at most $\ell$. In particular, note that, by Lemma~\ref{tree_vctree_lemma},
$\mfG_\ell$ contains all graphs with treewidth at most $\ell-1$. 
Then, for an $[n,k,d]$ linear code $\cC$ with $n > 1$, we have, 
via (\ref{kmfG_def}), Corollary~\ref{cor1} and Proposition~\ref{treewidth_bnd},
$$
\k(\cC;\mfG_\ell) \geq \frac{\k_\tree(\cC)}{\ell}
\geq \frac{k(d-1)}{\ell \, n \, (3 + 2\log_2(n-1))}.
$$
We thus have the following corollary to Proposition~\ref{treewidth_bnd},
which extends Corollary~\ref{cor3}.

\begin{corollary}
Given an integer $\,t > 0$, if $\mfC$ is a family of codes over $\F$ 
with the property that $\k(\cC;\mfG_\ell) \leq t$ for all $\cC \in \mfC$,
then $\mfC$ is not asymptotically good.
\label{cor4}
\end{corollary}

A rough interpretation of the last result is that 
a ``good'' error-correcting code cannot have a low-complexity 
realization on a graph with small (vc-)treewidth. In other words,
codes that are good from an error-correcting standpoint can 
have low-complexity realizations only on graphs with large (vc-)treewidth.

\section{Concluding Remarks\label{conclusion}}

In this paper, we demonstrated at length the use of the Vertex-Cut Bound
in finding lower bounds on the $\k$-complexity of graphical realizations.
As suggested at the end of Section~\ref{cutset_bnd_section}, a natural
application of the Vertex-Cut Bound is to formulate constrained optimization
problems whose solutions are lower bounds to various measures of constraint 
complexity. This method could be used, for instance, to determine lower
bounds on the true computational complexity of sum-product decoding
for a given code. This can be done by choosing a constraint complexity
measure that accurately reflects the cost of sum-product decoding. 
Such a complexity measure can be chosen based on detailed counts 
of the number of arithmetic operations required in the sum-product algorithm;
see, for example, \cite{AM00},\cite{For01},\cite{For03}.

It is now probably fair to say that the main open problem in the area of 
graphical realizations of codes is to explicitly construct,
for a given code $\cC$ and an arbitrary graph $\cG$, 
realizations of $\cC$ on $\cG$ whose constraint complexity 
is within striking distance of the lower bounds found using 
the methods of this paper. While we have shown that our bounds 
can be tight for specific examples of codes $\cC$ and graphs $\cG$, 
their tightness in the generic instance remains to be investigated.


\appendix

\section{An Example\label{example_app}}

We provide here an example of a code $\cC$ and a graph decomposition 
$(\cG,\omega)$ of $\cC$, for which there is no realization 
$\G \in \mfR(\cC;\cG,\omega)$ such that 
$\k(\G) = \k(\cC;\cG,\omega)$ and $\k^+(\G) = \k^+(\cC;\cG,\omega)$.
The example we give is based on Example~3.1 in \cite{KV03}.

Consider the $[11,3,3]$ binary linear code $\cC$ generated by the codewords
$00011100000$, $00000111000$ and $00101010100$. We take 
$I = \{1,2,3,\ldots,11\}$ to be the index set of $\cC$, 
identifying the index $i$ with the $i$th coordinate of $\cC$. 
Now, let $\cG$ be the 11-cycle, and 
$\omega:I \to V(\cG)$ the index mapping depicted in Figure~\ref{11-cycle}. 
Figure~\ref{tailbiting_fig1} shows a tailbiting trellis realization, 
$\G_1$, of $\cC$ that extends $(\cG,\omega)$. Note that 
$\k(\G_1) = 2$, while $\k^+(\G_1) = 14$. A second tailbiting trellis
realization, $\G_2$, in $\mfR(\cC;\cG,\omega)$ is shown 
in Figure~\ref{tailbiting_fig2}, with $\k(\G_2) = 3$ and 
$\k^+(\G_2) = 13$. It can be shown (using arguments similar to those 
given in Example~3.1 in \cite{KV03}) that $\k(\G_1) = \k(\cC;\cG,\omega)$,
while $\k^+(\G_2) = \k^+(\cC;\cG,\omega)$, and that there is
no realization $\G \in \mfR(\cC;\cG,\omega)$ that simultaneously
achieves $\k(\cC;\cG,\omega)$ and $\k^+(\cC;\cG,\omega)$.

\begin{figure}
\epsfig{file=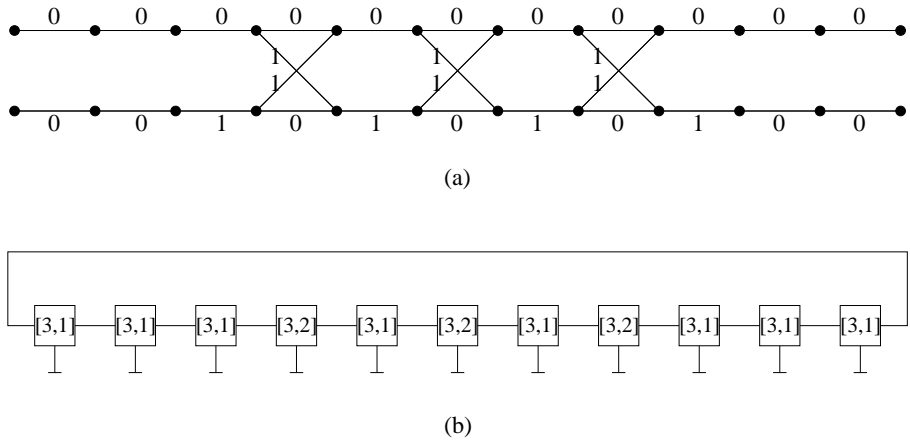, width=12cm}
\caption{A tailbiting trellis realization for an $[11,3,3]$ binary
linear code. (a)~Trellis diagram. (b)~Depiction of trellis realization 
showing the length and dimension of the local constraint codes; all
state spaces have dimension 1.}
\label{tailbiting_fig1}
\end{figure}

\begin{figure}
\epsfig{file=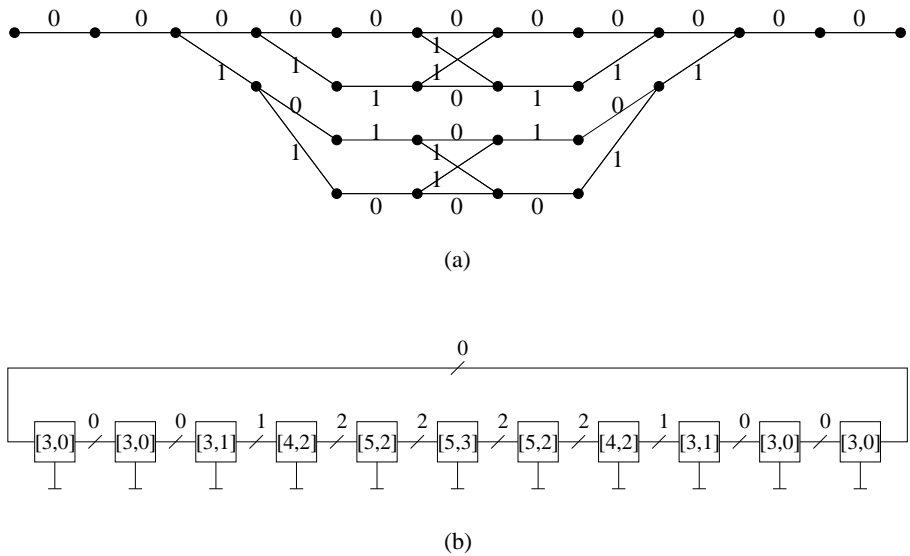, width=12cm}
\caption{A tailbiting trellis realization for an $[11,3,3]$ binary
linear code. (a)~Trellis diagram. (b)~Depiction of trellis realization 
showing  the length and dimension of the local constraint codes, and the
dimensions of the state spaces.}
\label{tailbiting_fig2}
\end{figure}


\begin{thebibliography}{99}

\bibitem{AM00} S.M.\ Aji and R.J.\ McEliece, 
``The generalized distributive law,'' 
\emph{IEEE Trans.\ Inform.\ Theory}, vol.\ 46, no.\ 2, pp.\ 325--343, 2000.

\bibitem{arnborg} S.\ Arnborg, D.G.\ Corneil and A.\ Proskurowski,
``Complexity of finding embeddings in a $k$-tree,''
\emph{SIAM J.\ Alg.\ Disc.\ Meth.}, vol.\ 8, pp.\ 277--284, 1987.

\bibitem{bod93} H.L.\ Bodlaender, ``A tourist guide through treewidth,''
\emph{Acta Cybernetica}, vol.\ 11, pp.\ 1--23, 1993.

\bibitem{BK96} H.L.\ Bodlaender, T.\ Kloks, 
``Efficient and constructive algorithms for the pathwidth and
treewidth of graphs,'' \emph{J.\ Algorithms}, vol.\ 21, no.\ 2,
pp.\ 358--402, 1996.

\bibitem{CFV99} A.R.\ Calderbank, G.D.\ Forney Jr., and A.\ Vardy, 
``Minimal tail-biting trellises: The Golay code and more,'' 
\emph{IEEE Trans.\ Inform.\ Theory}, vol.\ 45, pp.\ 1435--1455, July 1999.

\bibitem{For01} G.D.\ Forney Jr.,
``Codes on graphs: normal realizations,''
\emph{IEEE Trans.\ Inform.\ Theory}, 
vol.\ 47, no.\ 2, pp.\ 520--548, Feb.\ 2001.

\bibitem{For03} G.D.\ Forney Jr.,
``Codes on graphs: constraint complexity of cycle-free
realizations of linear codes,'' 
\emph{IEEE Trans.\ Inform.\ Theory}, 
vol.\ 49, no.\ 7, pp.\ 1597--1610, July 2003.


\bibitem{halford} T.R.\ Halford and K.M.\ Chugg,
``The extraction and complexity limits of graphical models
for linear codes,'' \emph{IEEE Trans.\ Inform.\ Theory}, 
to appear.

\bibitem{HW06} P.\ Hlin{\v e}n\'y and G.\ Whittle, ``Matroid tree-width,'' 
\emph{Europ.\ J.\ Combin.}, vol.\ 27, pp.\ 1117--1128, 2006.

\bibitem{jensen} F.V.\ Jensen, \emph{An Introduction to Bayesian Networks},
Springer, New York, 1996.

\bibitem{K2} N.\ Kashyap, 
``On minimal tree realizations of linear codes,''
submitted to \emph{IEEE Trans.\ Inform.\ Theory}. 
ArXiv e-print 0711.1383

\bibitem{KFL01} F.R.\ Kschischang, B.J.\ Frey and H.-A.\ Loeliger,
``Factor graphs and the sum-product algorithm,'' 
\emph{IEEE Trans.\ Inform.\ Theory},
vol.\ 47, no.\ 2, pp.\ 498--519, Feb.\ 2001.

\bibitem{KV02} R.\ Koetter and A.\ Vardy, 
``On the theory of linear trellises,'' 
in \emph{Information, Coding and Mathematics}, 
M.\ Blaum, P.G.\ Farrell and H.C.A.\ van Tilborg, eds., 
Kluwer, Boston, Mass., May 2002, pp.\ 323--354.

\bibitem{KV03} R.\ Koetter and A.\ Vardy, 
``The structure of tail-biting trellises: minimality and basic principles,'' 
\emph{IEEE Trans.\ Inform.\ Theory}, 
vol.\ 49, no.\ 9, pp.\ 2081--2105, Sept.\ 2003.

\bibitem{LV95} A.\ Lafourcade and A.\ Vardy, 
``Asymptotically good codes have infinite trellis complexity,'' 
\emph{IEEE.\ Trans.\ Inform.\ Theory}, 
vol.\ 41, no.\ 2, pp.\ 555--559, March 1995.


\bibitem{RS-I} N.\ Robertson and P.D.\ Seymour, 
``Graph minors. I. Excluding a forest,''
\emph{J.\ Combin.\ Theory, Ser.\ B}, vol.\ 35, pp.\ 39--61, 1983.


\bibitem{vardy} A.\ Vardy, ``Trellis Structure of Codes,'' 
in \emph{Handbook of Coding Theory}, R.\ Brualdi, C.\ Huffman and V.\ Pless,
Eds., Amsterdam, The Netherlands: Elsevier, 1998.

\bibitem{wiberg} N.\ Wiberg, \emph{Codes and Decoding on General Graphs},
Ph.D.\ thesis, Link\"oping University, Link\"oping, Sweden, 1996.

\bibitem{WLK95} N.\ Wiberg. H.-A.\ Loeliger and R.\ Koetter, 
``Codes and iterative decoding on general graphs,'' 
\emph{Euro.\ Trans.\ Telecommun.}, vol.\ 6, pp.\ 513--525, Sept./Oct.\ 1995.

\end{thebibliography}
\end{document}